%% file: itrs10.tex
\nonstopmode

\newcommand{\Comment}[1]{}

\documentclass{eptcs}

\newskip \point \point=1.1pt
\newif\ifmycolour \mycolourfalse
\def\firstitem{\item}
\def \shiftitem {\item}

\usepackage{theorems}
\usepackage{mathpazo}
\usepackage{times}
\usepackage{breakurl}

\def\titlerunning{Sound and Complete Typing for $\lmu$}
\def\authorrunning{Steffen van Bakel}

 \author{Steffen van Bakel 
\institute{Department of Computing, Imperial College London, 180 Queen's Gate, London SW7 2BZ, UK 
\\
 \email{s.vanbakel@imperial.ac.uk}}
}

\def\proof{ \trivlist\item[\hskip \labelsep {\it Proof.}] \ignorespaces \rm}

 \input{Xmacros}

 \title {Sound and Complete Typing for $\lmu$}
 \date{}

\begin{document}
 \bibliographystyle{eptcs}

\maketitle

 \begin{abstract}
In this paper we define intersection and union type assignment for Parigot's calculus $\lmu$.
We show that this notion is complete (\emph{i.e.}~closed under subject-expansion), and show also that it is sound (\emph{i.e.}~closed under subject-reduction).
This implies that this notion of intersection-union type assignment is suitable to define a semantics.
 \end{abstract}


 \section*{Introduction}
The Intersection Type Discipline has proven to be an expressive tool for studying termination and semantics for the $`l$-calculus~\cite{Church'36,Barendregt'84}.
Intersection type assignment is defined as an extension of the standard, implicative type assignment known as Curry's system \cite{Curry-Feys'58} (see also \cite{Hindley'97}), which expresses function composition and application; the extension made consists of relaxing the requirement that a parameter for a function should have a single type, adding the type constructor $\inter$ next to $\arrow$.
This simple extension allows for a great leap in complexity: not only can a (filter) model be built for the {\LC} using intersection types, also strong normalisation (termination) can be characterised via assignable types; however, type assignment becomes undecidable.
The literature on intersection types is vast; it was first defined by Coppo and Dezani-Ciancaglini in \cite{Coppo-Dezani'78} and its development took place over a number of years, culminating in 
the paper by Barendregt, Coppo, and Dezani-Ciancaglini \cite{BCD'83}, and has been explored by many people since.

It is natural to ask if these results can be achieved for other calculi (reduction systems) as well, and in previous papers the author investigated (in collaboration) Term Rewriting Systems \cite{Bakel-Fernandez-IaC'97}, and Object Oriented Calculi \cite{Bakel-deLiguoro-TOCS'08}; Maffeis looked at intersection types in the context of the $`p$-calculus \cite{Maffeis'05}.
In order to come to a characterisation of strong normalisation for Curien and Herbelin's (untyped) sequent calculus $\lmmt$ \cite{Curien-Herbelin'00}, Dougherty, Ghilezan and Lescanne presented System $\MIU$ \cite{DGL-ITRS'04}, that defines a notion of intersection and union typing for that calculus; in a later paper \cite{DGL-CDR'08}, they presented an improved version of their original system.

In \cite{Bakel-APAL'10}, the author revisited System $\MIU$, and showed that that system was neither \emph{sound} (\emph{i.e.}~closed under reduction), nor \emph{complete} (\emph{i.e.}~closed under reverse reduction); the same holds for the system presented in \cite{DGL-CDR'08}.
To address completeness, \cite{Bakel-APAL'10} adds $\Top$ as the maximal and $\Bottom$ as the minimal type, and extends the set of derivation rules; however, soundness is shown to be impossible to achieve without restricting typeability (effectively making less terms typeable).
In \cite{Bakel-FI'10}, the author attempted to solve the same issue, but this time in the context of the sequent calculus $\X$, as defined by Lengrand \cite{Lengrand'03}, and later studied by Lescanne and the author \cite{vBLL-ICTCS'05,Bakel-Lescanne-MSCS'08}; $\X$ is a sequent calculus in that it enjoys the Curry-Howard isomorphism with respect to the implicative fragment of Gentzen's $\LK$ \cite{Gentzen'35}.
The advantage of using the sequent approach is that it is now possible to explore the duality of intersection and union fully, through which we can study and explain various anomalies of union type assignment \cite{Pierce-PhD'91,Barbanera-Dezani-Liguoro-IaC'95} and quantification \cite{Harper-Lillibridge'91,Milner-et.al'97}.
Also for $\X$, the completeness result follows relatively easily, but soundness can only be shown for restricted systems (effectively call-by-name and call-by-value reduction, but it might be possible that other sound restrictions exist as well).
The main conclusion of those papers is that, in symmetric calculi (like $\lmmt$ and $\X$) it is inevitable that intersection and union are truly dual, and that the very nature of those calculi makes a sound and complete system unachievable.

In this paper we will continue on this path and bring intersection types to the context of \emph{classical logic}, by presenting a notion of intersection and union type assignment for the (untyped) calculus $\lmu$, that was first defined by Parigot in \cite{Parigot'93}, and was later extensively studied by various authors.

Intersection and union types have also been studied in the context of the {\LC} in \cite{Barbanera-Dezani-Liguoro-IaC'95}; also for the system defined in that paper soundness is lost, which can only be recovered by limiting to parallel reduction, \emph{i.e.}~all residuals of a redex need to be contracted in parallel.
The problem of loss of soundness also appears in other contexts, such as that of {\ML} with side-effects \cite{Harper-Lillibridge'91,Wright'95,Milner-et.al'97}, and that of using intersection and union types in an operational setting \cite{Davies-Pfenning'01,Dunfield-Pfenning'03}.
As here, also there the cause of the problem is that the type-assignment rules are not fully logical, making the context calls (which form part of the reduction in $\X$) unsafe; this has, in part, already been observed in \cite{Herbelin'05} in the context of Curien and Herbelin's calculus $\lmmt$ \cite{Curien-Herbelin'00}.
This also explains why, for {\ML} with side-effects, quantification is no longer sound \cite{Harper-Lillibridge'91,Milner-et.al'97}: also the $(\forall I)$ and $(\forall E)$ rules of {\ML} are not logical.

In the view of those failures, the result presented here comes as a surprise.
We will define a notion of type assignment for $\lmu$ that uses intersection and union types, and show that it is \emph{both} sound \emph{and} complete.
The system presented is a natural extension of the strict intersection type assignment system as defined in \cite{Bakel-TCS'92}; this implies that intersection models the distribution of arguments in a parameter call.
But it is also a natural extension of the system for $\lmu$, and in order to achieve completeness for structural reduction, as in the papers mentioned above, union types are added.
However, the union types are no longer dual to intersection types; union types play only a marginal role, as was also the intention of \cite{DGL-CDR'08}.
Contrary to that paper, however, we do not see union as negated intersection, but see a union type as a strict type; in particular, we do not allow the normal $(\unI)$ and $(\unE)$ rules as used in \cite{Barbanera-Dezani-Liguoro-IaC'95}, which we know create the same soundness problem.
Moreover, although one can link intersection types with the logical connector \textsl{and}, the union types we use here have \emph{no} relation with \textsl{or}; one could argue that therefore perhaps \emph{union} is not the right name to use for this type constructor, but we will stick with it nonetheless.

The limited view of union types is mirrored by $\lmu$'s limited (with respect to $\lmmt$ and $\X$) notion of context\footnote{In particular, $\lmmt$'s $\mt x.c$ is not represented.}.
In $\lmu$, we distinguish \emph{control structures} as those terms that start with a context switch $`m`a.[`b]M$, followed by a number of arguments; since union types allow us to express that the various continuations (all called $`a$) need not have the same type, we use a different formulation for rule $(\arrE)$, which has an implicit use of union elimination (see Definition~\ref{TurnLmuIU system}). 
The type system defined here will be shown to be the natural one, in that intersection and union play their expected roles for completeness.
Because the use of intersection and union is limited in that a context variable cannot have an intersection type, and although we allow union types for term variables, we do not have the normal union elimination rule; thanks to these two restrictions, we can show soundness as well.

 \section{The calculus $`l`m$} \label{lmu}

Parigot's $`l`m$-calculus~\cite{Parigot'92} is a proof-term syntax for classical logic, expressed in Natural Deduction, defined as an extension of the Curry type assignment system for the {\LC}.
We quickly revise some basic notions:

 \begin{definition} [Lambda terms and $\beta$-contraction \cite{Barendregt'84}]
 \label {lambda terms}

 \begin{enumerate}

 \firstitem \label {L-terms}
\emph{$`l$-terms} are defined by:
 \[\begin{array}{rcl}
M,N &::=& x \mid `l x . M \mid MN
 \end{array} \]

 \item \label {red}
The reduction relation $\bred$ is defined as the contextual closure of the rule:
 \[\begin{array}{rcl}
( `l x . M ) N &\bred& M [ N /x ]
 \end{array} \]

 \end{enumerate}
 \end{definition}

Curry (or simple) type assignment for the $`l$-calculus is defined as:

 \begin{definition} \label {Curry types} \label {Curry type assignment}
 \begin{enumerate}

 \firstitem
Let $`v$ range over a countable (infinite) set of type-variables.
The set of \emph{Curry-types} is defined by the grammar:
 \[ \begin{array}{rcl}
A,B &::=& `v \mid (A\arr B)
 \end{array} \]

 \item
\emph{Curry-type assignment} is defined by the following natural deduction system.

 \[ \begin{array}{rl@{\dquad}rl@{\dquad}rl}
(\Ax): &
\Inf	{\der `G,x{:}A |- x : A } 
&
(\arrI): &
\Inf	{\der `G,x{:}A |- M : B }
	{ \der `G |- `lx.M : A\arr B }
&
(\arrE): &
\Inf	{ \der `G |- M : A\arr B \quad \der `G |- N : A }
	{ \der `G |- MN : B }
 \end{array} \]

 \end{enumerate}
 \end{definition}

With $`l`m$ Parigot created a multi-conclusion typing system which corresponds to classical logic; the derivable statements have the shape $\derLmu `G |- M : A | `D $, where $A$ is the main conclusion of the statement, expressed as the \emph{active} conclusion, and $`D$ contains the alternative conclusions, consisting of pairs of Greek characters and types; the left-hand context $`G$, as usual, contains pairs of Roman characters and types, and represents the types of the free term variables of $M$.
As with Implicative Intuitionistic Logic, the reduction rules for the terms that represent the proofs correspond to proof contractions; the difference is that the reduction rules for the {\LC} are the \emph{logical} reductions, \emph{i.e.}~deal with the elimination of a type construct that has been introduced directly above.
In addition to these, Parigot expresses also the \emph{structural} rules, where elimination takes place for a type constructor that appears in one of the alternative conclusions (the Greek variable is the name given to a subterm): he therefore needs to express that the focus of the derivation (proof) changes, and this is achieved by extending the syntax with two new constructs $[`a]M$ and $`m`a.M$ that act as witness to \emph{deactivation} and \emph{activation}, which together move the focus of the derivation.

\Comment{\section*{Variants of $\lmu$}}

We will now present the variant of $`l`m$ we consider in this paper, as considered by Parigot in \cite{Parigot-Brno'93}; for convenience, we split terms into two categories: we define \emph{terms}, and \emph{control structure}:

 \begin{definition}[Syntax of $\lmu$] \label{lm-terms}
The $`l`m$-\emph{terms} we consider are:
 \[ \begin{array}{r@{~}c@{~}l}
M,N &::=& x \mid `l x.M \mid MN \mid `m`a.[`b]M\,.
 \end{array} \]
We also define \emph{control structure} as a subset of terms:
 $ \begin{array}{r@{~}c@{~}l}
C &::=& `m`a.[`b]M \mid CM \, .
 \end{array} $
 \end{definition}
To shorten proofs and notation, we will treat $`m`a.M$ as a term as well, whenever convenient.

As usual, $`l x.M$ binds $x$ in $M$, and $`m `a .[`b]M$ binds $`a $ in $M$, and the notions of free and bound variables are defined accordingly; the notion of $`a $-conversion extends naturally to bound names, and we assume Barendregt's convention on free and bound variables.

In $\lmu$, reduction of terms is expressed via implicit substitution; as usual, $M[N/x]$ stands for the substitution of all occurrences of $x$ in $M$ by $N$, and $M [ N{`.}`g/`a ]$ stands for the term obtained from $M$ in which every (pseudo) sub-term of the form $[`a]M'$ is substituted by $[`g](M' N)$ ($`g $ is a fresh variable) (in Parigot's notation: $ (`m `a . [`b] M ) N \red `m`g . [`b] M [\, [`g]PN / [`a]P \,]$).

We define formally how to preform the $`m$-substitution; this is convenient in later proofs.

 \begin{definition}
We define $M [ N{`.}`g/`a ]$ by induction over the structure of terms by:
 \[ \begin{array}{r@{\,}lcll}
x & [ N{`.}`g/`a ] & \ByDef & x
\\
(`lx.M) & [ N{`.}`g/`a ] & \ByDef & `lx . (M \, [ N{`.}`g/`a ] )
\\
(M_1M_2) & [ N{`.}`g/`a ] & \ByDef & M_1 \, [ N{`.}`g/`a ]~M_2 \, [ N{`.}`g/`a ]
 \\
(`m`d.[`a]M) & [ N{`.}`g/`a ] & \ByDef & `m`d.[`g](M\,[ N{`.}`g/`a ] N)
\\
(`m`d.[`b]M) & [ N{`.}`g/`a ] & \ByDef & `m`d.[`b] ( M \, [ N{`.}`g/`a ] ) & `b \not= `a 
 \end{array} \]
 \end{definition}

We have the following rules of computation in $\lmu$:

 \begin{definition}[$`l`m$ reduction]
Parigot defines a number of reduction rules: two \emph{computational rules}

 \[ \begin{array}{rrcl}
\textit{logical } (\beta): & (`l x . M ) N & \red & M [ N /x ]
\\
\textit{structural } (\mu): & (`m `a . [`b] M ) N & \red & `m`g.( [`b] M[N{`.}`g/`a]) \hspace*{5.5cm}
 \end{array} \]
as well as the \emph{simplification rules}:
 \[ \begin{array}{rrcll}
\textit{renaming} : & `m`a[`b](`m`g.[`d]M) & \red & `m`a.[`d]M[`b/`g]
\\
\textit{erasing} : & `m `a . [`a] M & \red & M & \textit{if $`a$ does not occur in $M$.}
\\
`h`m : & `m`a . [`b] M & \red & `lx`m`g . [`b] M[x{`.}`g/`a]
 \end{array} \]
which are added mainly to simplify the presentation of his results%
\footnote{In fact, Parigot formulates the renaming rule as $ [`b](`m`g.M) \red M[`b/`g] $; since $[`b](`m`g.M)$ is not a term, we write the rule differently.}.

Reduction on $`l`m$-terms is defined as the compatible closure of these rules.
 \end{definition}
It is possible to define more reduction rules, but Parigot refrained from that since he aimed at defining a confluent reduction system. 
 
The intuition behind the structural rule is given by de Groote \cite{deGroote'94}: ``\emph{in a $\lmu$-term $`m`a.M$ of type $A \arr B$, only the subterms named by $`a$ are \emph{really} of type $A \arr B$ (\ldots); hence, when such a $`m$-abstraction is applied to an argument, this argument must be passed over to the sub-terms named by $`a$.}''
In this paper, we will only deal with the logical, structural and renaming rule; this is also the restriction made by de Groote in \cite{deGroote'94}.

Type assignment for $\lmu$ is defined by the following natural deduction system; there is a \emph{main}, or \emph{active}, conclusion, labelled by a term of this calculus, and the alternative conclusions are labelled by the set of Greek variables $`a ,`b, \textit{etc}$.

 \begin{definition}[Typing rules for $\lmu$]
Our types are those of Definition~\ref{Curry types}, extended with the type constant $\bot$ that is essentially added to express negation, \emph{i.e.}:
 \[ \begin{array}{rcll}
A,B &::=& `v \mid \bot \mid (A\arr B) & (A \not= \bot)
 \end{array} \]
The type assignment rules are:
\[ \begin{array}{rl@{\dquad}rl@{\dquad}rl}
 (\Ax) : &
 \Inf	[x{:}A\in`G]{ \derLmu `G |- x : A | `D }
&
 (`m) : &
 \Inf	{ \derLmu `G |- M : B | `a{:}A,`b{:}B,`D 
	}{ \derLmu `G |- `m`a.[`b]M : A | `b{:}B,`D }
\quad
 \Inf	{ \derLmu `G |- M : A | `a{:}A,`D 
	}{ \derLmu `G |- `m`a.[`a]M : A | `D }
\end{array} \]
\[ \begin{array}{rl@{\dquad}rl}
 (\arrI) : &
 \Inf	{ \derLmu `G,x{:}A |- M : B | `D 
	}{ \derLmu `G |- `l x.M : A\arrow B | `D }
&
 (\arrE) : &
 \Inf	{ \derLmu `G |- M : A\arrow B | `D 
	 \quad
	 \derLmu `G |- N : A | `D 
	}{ \derLmu `G |- MN : B | `D }
\end{array} \]
 \end{definition}

We can think of $[`a]M$ as storing the type of $M$ amongst the alternative conclusions by giving it the name $`a $ - the set of Greek variables is called the set of \emph{context} variables (or \emph{names}).

\Comment
{
Note that we have the \emph{Weakening Property}: If $ \DerLmu { `G }{ M}{A}{`D }$ and $`G \subseteq`G'$ and $`D \subseteq`D'$, then $ \DerLmu { `G' }{ M}{A}{`D' }$.

Since $\bot$ is a proper type, we can express negation, and can inhabit $((A\arr \bot) \arr \bot) \arr A$ with $`ly.`m`a.[`b]y(`lx.`m`d.[`a]x)$; notice that this term is not closed, since $`b$ is free.
We can inhabit Pierce's law with a closed term: $ \derLmu { } |- `lx.`m`a.[`a](x(`ly.`m`b.[`a]y)) : ((A\arrow B)\arrow A)\arrow A | {} $; the underlying logic for this system corresponds to \emph{minimal classical logic} \cite{Ariola-Herbelin'03}.

 \[ \begin{array}{c}
\Inf	[\arrI]
	{\Inf	[`m]
		{\Inf	[\arrE]
			{\Inf	[\Ax]
				{\derLmu y{:}\neg\neg A |- y : \neg\neg A | {} }
			 \Inf	[\arrI]
				{\Inf	[`m]
					{\Inf	[\Ax]
						{\derLmu x{:}A |- x : A | `a{:}A,`d{:}{\bot} }
					}{\derLmu x{:}A |- `m`d.[`a]x : {\bot} | `a{:}A }
				}{\derLmu {} |- `lx.`m`d.[`a]x : \neg A | `a{:}A }
			}{\derLmu y{:}\neg\neg A |- [`b]y(`lx.`m`d.[`a]x) : A | `a{:}A,`b{:}A }
		}{\derLmu `m`a.y{:}\neg\neg A |- `m`a.[`b]y(`lx.`m`d.[`a]x) : A | `b{:}A }
	}{\derLmu {} |- `ly.`m`a.[`b]y(`lx.`m`d.[`a]x) : \neg\neg A\arr A | `b{:}A }
 \end{array} \]

 \[ \begin{array}
\Inf	[\arrI]
	{\Inf	[`m]
		{\Inf	[\arrE]
			{\Inf	[\Ax]
				{ \derLmu x{:}(A\arrow B)\arrow A |- x : (A\arrow B)\arrow A | `a{:}A }
			 \Inf	[\arrI]
			 	{\Inf	[`m]
					{\Inf	[\Ax]
						{ \derLmu x{:}(A\arrow B)\arrow A,y{:}A |- y : A | `a{:}A,`b{:}B }
					}{ \derLmu x{:}(A\arrow B)\arrow A,y{:}A |- `m`b.[`a]y : B | `a{:}A }
				}{ \derLmu x{:}(A\arrow B)\arrow A |- `ly.`m`b.[`a]y : A\arrow B | `a{:}A }
			}{ \derLmu x{:}(A\arrow B)\arrow A |- x(`ly.`m`b.[`a]y) : A | `a{:}A }
		}{ \derLmu x{:}(A\arrow B)\arrow A |- `m`a.[`a](x(`ly.`m`b.[`a]y)) : A | {} }
	}{ \derLmu { } |- `lx.`m`a.[`a](x(`ly.`m`b.[`a]y)) : ((A\arrow B)\arrow A)\arrow A | {} }
 \end{array} \]

:
 \[ \begin{array}{c}
\Inf	[\arrI]
	{	\Inf	[`m]
			{\Inf	[\arrE]
				{\Inf	[\Ax]
					{\derLmu y{:}(A\arr\bot)\arr{\bot} |- y : (A\arr\bot)\arr{\bot} | {} }
				 \Inf	[\arrI]
					{	\Inf	[\bot]
							{\Inf	[\Ax]
								{\derLmu x{:}A |- x : A | `a{:}A }
							}{\derLmu x{:}A |- [`a]x : {\bot} | `a{:}A }
						}{\derLmu {} |- `lx.[`a]x : A\arr{\bot} | `a{:}A }
				}{\derLmu y{:}(A\arr\bot)\arr{\bot} |- y(`lx.[`a]x) : {\bot} | `a{:}A }
			}{\derLmu y{:}(A\arr\bot)\arr{\bot} |- `m`b.y(`lx.[`a]x) : A | `a{:}A }
	}{\derLmu {} |- `ly.`m`b.y(`lx.[`a]x) : ((A\arr\bot)\arr\bot)\arr A | {} }
 \end{array} \]
 notice that $`b$ is free in this proof.
Alternatively, as shown by de Groote \cite{deGroote'94}, using a more relaxed system that treats $[`a]M$ and $`m`a.M$ as terms, it is possible to use the term $ `ly.`m`b.y(`lx.[`a]x) $.

 \[ \begin{array}{c}
\Inf	[\arrI]
	{\Inf	[`m]
		{\Inf	[\bot]
			{\Inf	[\arrE]
				{\Inf	[\Ax]
					{\derLmu y{:}\neg\neg A |- y : \neg\neg A | {} }
				 \Inf	[\arrI]
					{\Inf	[`m]
						{\Inf	[\bot]
							{\Inf	[\Ax]
								{\derLmu x{:}A |- x : A | `a{:}A,`d{:}A }
							}{\derLmu x{:}A |- [`a]x : {\bot} | `a{:}A,`d{:}A }
						}{\derLmu x{:}A |- `m`d.[`a]x : A | `a{:}A }
					}{\derLmu {} |- `lx.`m`d.[`a]x : \neg A | `a{:}A }
				}{\derLmu y{:}\neg\neg A |- y(`lx.`m`d.[`a]x) : {\bot} | `a{:}A,`b{:}A }
			}{\derLmu y{:}\neg\neg A |- [`b]y(`lx.`m`d.[`a]x) : A | `a{:}A,`b{:}A }
		}{\derLmu `m`a.y{:}\neg\neg A |- `m`a.[`b]y(`lx.`m`d.[`a]x) : A | `b{:}A }
	}{\derLmu {} |- `ly.`m`a.[`b]y(`lx.`m`d.[`a]x) : \neg\neg A\arr A | `b{:}A }
 \end{array} \]

 \[ \begin{array}{c}
\Inf	[\arrI]
	{\Inf	[`m]
		{\Inf	[\arrE]
			{\Inf	[\Ax]
				{\derLmu y{:}\neg\neg A |- y : \neg\neg A | {} }
			 \Inf	[\arrI]
				{\Inf	[`m]
					{\Inf	[\Ax]
						{\derLmu x{:}A |- x : A | `a{:}A,`d{:}A }
					}{\derLmu x{:}A |- `m`d.[`a]x : A | `a{:}A }
				}{\derLmu {} |- `lx.`m`d.[`a]x : \neg A | `a{:}A }
			}{\derLmu y{:}\neg\neg A |- y(`lx.`m`d.[`a]x) : A | `a{:}A,`b{:}A }
		}{\derLmu `m`a.y{:}\neg\neg A |- `m`a.[`b]y(`lx.`m`d.[`a]x) : A | `b{:}A }
	}{\derLmu {} |- `ly.`m`a.[`b]y(`lx.`m`d.[`a]x) : \neg\neg A\arr A | `b{:}A }
 \end{array} \]
}

As an example illustrating the fact that this system is more powerful than the system for the \LC, here is a proof of Peirce's Law (due to Ong and Steward~\cite{Ong-Stewart'97}):
\[ 
 \Inf	[\arrI]
	{\Inf	[`m]
		{\Inf	[\bot]
			{\Inf	[\arrE]
				{\Inf	[\Ax]
					{ \derLmu x{:}(A\arrow B)\arrow A |- x : (A\arrow B)\arrow A | `a{:}A }
				 \ 
				 \Inf	[\arrI]
				 	{\Inf	[`m]
						{\Inf	[\bot]
							{\Inf	[\Ax]
								{ \derLmu x{:}(A\arrow B)\arrow A,y{:}A |- y : A | `a{:}A,`b{:}B }
							}{ \derLmu x{:}(A\arrow B)\arrow A,y{:}A |- [`a]y : {\bot} | `a{:}A }
						}{ \derLmu x{:}(A\arrow B)\arrow A,y{:}A |- `m`b.[`a]y : B | `a{:}A }
					}{ \derLmu x{:}(A\arrow B)\arrow A |- `ly.`m`b.[`a]y : A\arrow B | `a{:}A }
				}{ \derLmu x{:}(A\arrow B)\arrow A |- x(`ly.`m`b.[`a]y) : A | `a{:}A }
			}{ \derLmu x{:}(A\arrow B)\arrow A |- [`a](x(`ly.`m`b.[`a]y)) : {\bot} | `a{:}A }
		}{ \derLmu x{:}(A\arrow B)\arrow A |- `m`a.[`a](x(`ly.`m`b.[`a]y)) : A | {} }
	}{ \derLmu { } |- `lx.`m`a.[`a](x(`ly.`m`b.[`a]y)) : ((A\arrow B)\arrow A)\arrow A | {} }
\]
Notice that $\bot$ plays no part in this proof.
Indeed, we can define the set of types without $\bot$; the underlying logic of such a system then corresponds to \emph{minimal classical logic} \cite{Ariola-Herbelin'03}.

Since we allow $\bot$ as a proper type, we can even express negation (of course, it is also implicitly present in the right-hand type environment), and can give a derivation for $\neg\neg A \arr A$, so can express \emph{double negation elimination}; so in $\lmu$ we can represent full Classical Logic.

%

 \[ \begin{array}{c}
\Inf	[\arrI]
	{\Inf	[{\Red `m}]
		{\Inf	[\arrE]
			{\Inf	[\Ax]
				{\derLmu {\Stat y:(A\arr \bot)\arr \bot} |- y : {(A\arr \bot)\arr \bot} | \Stat`b:{\bot} }
			 \Inf	[\arrI]
				{\Inf	[{\Red `m}]
					{\Inf	[\Ax]
						{\derLmu \Stat x:A |- x : A | \Stat `a:A,\Stat`d:{\bot} }
					}{\derLmu \Stat x:A |- `m`d.[`a]x : {\bot} | \Stat`a:A }
				}{\derLmu {} |- `lx.`m`d.[`a]x : A\arr {\bot} | \Stat`a:A }
			}{\derLmu {\Stat y:(A\arr \bot)\arr \bot} |- y(`lx.`m`d.[`a]x) : {\bot} | \Stat`a:A,\Stat`b:{\bot} }
		}{\derLmu \Stat y:(A\arr \bot)\arr {\bot} |- `m`a.[`b]y(`lx.`m`d.[`a]x) : A | \Stat`b:{\bot} }
	}{\derLmu {} |- `ly.`m`a.[`b]y(`lx.`m`d.[`a]x) : ((A\arr \bot)\arr \bot)\arr A | \Stat`b:{\bot} }
 \end{array} \]
Notice that this term is not closed, since $\lterm{`b}$ is free, albeit of type $\bot$.

De Groote \cite{deGroote'94} considers a variant of $\lmu$ which separates the naming and $`m$-binding features\footnote{Notice that then Parigot's renaming rule is correct.
We could have presented our results for this more permissive system, but would have had to sacrifice soundness and completeness for the renaming rule.
Notice that we would still have soundness and completeness for the two computational rules, which are arguably the most important.}.
This gives a considerable different system, that allows for $\neg \neg A \arr A $ to be inhabited via (the closed term) $ `ly.`m`a.y(`lx.[`a]x) $.
De Groote's variant of $\lmu$ \cite{deGroote'94} uses the syntax 
 \[ \begin{array}{r@{~}c@{~}l}
\lterm M, \lterm N &::=& \lterm x \mid \lterm {`l x.M} \mid \lterm {MN} \mid \lterm {`m`a.M} \mid \lterm {[`b]M}
 \end{array} \]
and splits rule $({\Red `m})$ into
 \[ \begin{array}{rl@{\dquad}rl}
 ({\Red `m}) : &
 \Inf	{ \derLmu `G |- M : {\bot} | \Stat`a:A,`D 
	}{ \derLmu `G |- `m`a.M : A | `D }
&
 ({\Red \bot}) : &
 \Inf	{ \derLmu `G |- M : A | \Stat`b:A,`D 
	}{ \derLmu `G |- [`b]M : {\bot} | \Stat`b:A,`D }
 \end{array} \]

In this system we can derive
 \[ \begin{array}{c}
\Inf	[\arrI]
	{\Inf	[{\Red `m}]
		{\Inf	[\arrE]
			{\Inf	[\Ax]
				{\derLmu \Stat y:(A\arr\bot)\arr{\bot} |- y : (A\arr\bot)\arr{\bot} | {} }
			 \Inf	[\arrI]
				{\Inf	[{\Red \bot}]
					{\Inf	[\Ax]
						{\derLmu \Stat x:A |- x : A | \Stat`a:A }
					}{\derLmu \Stat x:A |- [`a]x : {\bot} | \Stat`a:A }
				}{\derLmu {} |- `lx.[`a]x : A\arr{\bot} | \Stat`a:A }
			}{\derLmu y{:}(A\arr\bot)\arr{\bot} |- y(`lx.[`a]x) : {\bot} | \Stat`a:A }
		}{\derLmu y{:}(A\arr\bot)\arr{\bot} |- `m`a.y(`lx.[`a]x) : A | {} }
	}{\derLmu {} |- `ly.`m`a.y(`lx.[`a]x) : ((A\arr\bot)\arr\bot)\arr A | {} }
 \end{array} \]

For the moment, we will deal with Parigot's original system only; we aim to extend our results to de Groote's variant in future work.

 \section{The Strict Intersection Type Assignment System for the $`l$-calculus} \label{strict int for LC}

The remainder of this paper will be dedicated a notion of intersection/union typing on $\lmu$.
This will be defined as a natural extension of the Strict Intersection System \cite{Bakel-TCS'92} for the \LC.
Before we come to that, we will briefly summarise the latter.

 \begin {definition} [Strict types] \label {strict types}
 \begin {enumerate}

 \firstitem \label {Tstrict}
Let $`v$ range over an infinite, enumerable set of type variables.
The set $\Tstrict$ of \emph{strict types}, ranged over by $ A, B,$ \textit{etc} is defined through the grammar:
 \[ \begin{array}{rlll}
 A, B & ::= & `v \mid \Top \arr B \mid (A_1 \inter \dots \inter A_n) \arr B & (n \geq 1)
 \end{array} \]

The set $\Types$ of \emph{intersection types} is defined as the union of $\{\Top\}$ and the closure of $\Tstrict$ under intersection; we will use $A,B,$ \textit{etc} for intersection types as well, and mention which set they belong to when necessary.

 \item A \emph{statement} is an expression of the form $\lstat{M}{A}$, with $M \ele `L$ and $ A \ele \Types$.
$M$ is the \emph{subject} and $A$ the \emph{predicate} of $\lstat{M}{A}$.

 \item A \emph{type-environment} $`G$ is a partial mapping from term variables to intersection types, and we write $x{:}A \ele `G$ if $`G\,(x) = A$.

 \end{enumerate}
 \end{definition}
So if we write a type as $A \arrow B$, then $B \ele \TStrict$, and $A \ele \Types$.

In the notation of types, as usual, right-most outer-most parentheses in arrow types will be omitted, and we assume $\inter$ to bind stronger than $\arrow$.
From hereon, we will write $\n$ for the set $\{1,\ldots,n\}$.

We will consider a pre-order on types which takes into account the idem-potence, commutativity and associativity of the intersection type constructor, and defines $\Top$ to be the maximal element.

 \begin {definition} \label {seq definition}
 \begin {enumerate}

 \firstitem 
The relation `$\leq$' is defined as the least pre-order on $\Types$ such that:
 \[ \begin{array}{rcll}
A_1 \inter \dots \inter A_n &\seq& A_i, & \mbox{ for all } \iotn, n \geq 1
\\
 B \seq A_i,\ \mbox{ for all } \iotn & \Then & B \seq A_1 \inter \dots \inter A_n, & n \geq 0
 \end{array} \]

 \item
On $\Types$, the relation `$\equ$' is defined by: 
 \[ \begin{array}{rcl@{\qquad}rcl}
 A \seq B \seq A & \Then & A \equ B 
& 
 A \equ B \And C \equ D & \Iff & A \arr C \equ B \arr D
 \end{array} \]

 \item
The relations `$\seq$' and `$\equ$', are extended to contexts by: 
$`G\seq `G'$ if and only if for every $x{:}A' \ele `G'$ there is an $x{:}A \ele `G$ such that $ A \seq A' $, and: $`G\equ `G'\Iff$ $`G\seq `G'\seq `G$.

 \end{enumerate}
 \end{definition}

$\Types$ will be considered modulo $\equ$; then $\seq$ becomes a partial order.
It is easy to show that both $(A \inter B) \inter C \sim A \inter (B \inter C)$ and $A \inter B \sim B \inter A$, so the type constructor $\inter$ is associative and commutative, and w%
e will write $\AoI{n}$ for $A_1 \inter \dots \inter A_n$, and consider $\Top$ to be the empty intersection: $\Top = \int{0}A_i$.
Moreover, we will assume, unless stated explicitly otherwise, that in $\AoIn$ each $A_i$ is strict.

 \begin {definition} \label {strict type assignment definition}
The \emph{strict type assignment} is defined by the following natural deduction system (where all types mentioned are strict, with the exception of $A$ in rule $(\arrI)$ and $(\arrE)$):
 \[ \begin {array}{rlcrl}
(\intE): &
 \Inf	[n \beq 1, \jotn]
	{ \der `G,x{:}{\AoIn} |- x : A_j }
&&
(\intI):&
 \Inf [ n \beq 0 ]
	{ \der `G |- M : A_i \quad (\forall \iotn)
 	}{ \der `G |- M : {\AoIn} }
\\[5mm]
(\arrI):&
 \Inf	{ \der `G, x{:}A |- M : B }
	{ \der `G |- `l x.M : A\arr B }
&&
(\arrE):&
\Inf	{ \der `G |- M : A\arr B \quad \der `G |- N : A }
	{ \der `G |- MN : B }
 \end {array} \]
We will write $\der `G |- M : A $ for statements that are derived using these rules.

 \end {definition}
Notice that $\der `G |- M : {\Top} $ for all $`G,M$ by rule $(\intI)$.

Properties of this system have been studied in \cite{Bakel-TCS'92}.

\section{Intersection and union type assignment for $`l`m$}
We will now define a notion of type assignment for $\lmu$ that uses intersection and union types.

We see the context variables $`a$ as names for possible continuations that in the philosophy of intersection types need not all be typed with the same type; we therefore allow multiple types for a context variable in the environment $`D$, grouped using a new type constructor, which we call union.

Binding a context variable then generates a \emph{context switch} $`m`a.[`b]M$, which naturally has a union type $\AoUn$; reduction of the term $(`m`a.[`b]M)N$ will bring the operand $N$ to each of the pseudo subterms in $M$ of the shape $[`a]Q$ (`named' with $`a$), where $Q$ has type $A_i$; since $N$ gets placed behind $Q$, this implies that $A_i = C_i \arr B_i$ and that therefore the type for $`a$ should be $\un{n}(C_i\arr B_i)$; this then also implies that $N$ should have all the types $\Cotn$; rule $(\arrE)$ as below expresses exactly that.
The only `functionality' we need for union types therefore is the ability to choose a collection of types for $`a$ amongst those stored in $`D$; this is represented by rule $(\unE)$.

 \begin{definition}[The system $\TurnLmuIU$] 
 \begin{enumerate}
 
 \firstitem
The set of strict types we consider for the intersection-union type assignment is:
 \[ \begin{array}{rlll}
A, B & ::= & `v \mid B_1 \union \dots \union B_m \mid (A_1 \inter \dots \inter A_n) \arr B & (n,m \geq 0)
 \end{array} \]
As above, we call $A_1 \inter \dots \inter A_n$ (with $n \geq 0$) an intersection type, and call $B_1 \union \dots \union B_m$ (with $m \geq 0$) a \emph{union} type; we use $\Top$ for the empty intersection type, and $\Bottom$ for the empty union type.

 \item The relation $\seq$ of Definition~\ref {seq definition} is extended to intersection-union types by:
 \[ \begin{array}{rcll}
A_1 \inter \dots \inter A_n &\seq& A_i, & \mbox{for all } \iotn, n \geq 1
\\
 B \seq A_i, \mbox{ for all } \iotn & \Then & B \seq A_1 \inter \dots \inter A_n, & n \geq 0 \\
 B_j &\seq& B_1 \union \dots \union B_m, & \mbox{for all } \jotm , m \geq 1 \\
 B_j \seq A, \mbox{ for all } \jotm & \Then & B_1 \union \dots \union B_m \seq A, & m \geq 0 \\
 \end{array} \]

On $\Types$, the relation `$\equ$' is defined by the same way as in Definition~\ref {seq definition}.

 \item
A \emph{left environment} $`G$ is a partial mapping from term variables to intersections of strict types, and we write $x{:}A \ele `G$ if $`G\,(x) = A$.
Similarly, a \emph{right environment} $`D$ contains only strict types, which can be union types.

 \item
The relations `$\seq$' and `$\equ$', are extended to left and right environments by: 
$`G\seq `G'$ if and only if for every $x{:}A' \ele `G'$ there is an $x{:}A \ele `G$ such that $ A \seq A' $, and $`G\equ `G'\Iff$ $`G\seq `G'\seq `G$, and $`D \seq `D'$ if for every $`a{:}A \ele `D$ there exists $`a{:}A' \ele `D'$ such that $A \seq A'$, and $`D \equ `D' \Iff `D \seq `D' \seq `D$.

 \end{enumerate}
 \end{definition}
Notice that we consider union types to be strict as well; this implies that we allow an intersection of union types, a union of union types, but not a union of intersection types.

 \begin{definition}[The system $\TurnLmuIU$] \label{TurnLmuIU system}
 
Intersection-union type assignment for $`l`m$ is defined via: 

\noindent
{ \def\TurnLmuIU{\Turn}
 \[ \kern-5mm
 \begin{array}{rl@{\dquad}rl}
 (\intE): &
 \Inf	{ \derLmuIU `G,x{:}\AoI{n} |- x : A_i | `D }
&
(\intI): &
\Inf	[n \geq 0, n \not= 1]
	{ \derLmuIU `G |- M : A_i | `D \quad (\forall \iotn)
	}{ \derLmuIU `G |- M : \AoI{n} | `D }
\\[6mm]
(\arrI): &
 \Inf	{ \derLmuIU `G, x{:}A |- M : B | `D 
	}{ \derLmuIU `G |- `l x.M : A\arrow B | `D }
&
 (\arrE): &
 \Inf	[n \geq 1]
	{ \derLmuIU `G |- M : \un{n}(A_i\arrow B_i) | `D 
	 \quad
	 \derLmuIU `G |- N : A_i | `D 
	 ~(\forall \iotn)
	}{ \derLmuIU `G |- MN : \BoU{n} | `D }
 \\[6mm]
  (\unE): &
\multicolumn{3}{l}{
 \Inf	[\BoUm \seq \AoUn]
	{ \derLmuIU `G |- M : \BoU{m} | `b{:}\AoUn,`a{:}B,`D 
	}{ \derLmuIU `G |- `m`a.[`b]M : B | `b{:}\AoUn,`D }
\dquad
 \Inf	[\BoUm \seq \AoUn]
	{ \derLmuIU `G |- M : \BoU{m} | `b{:}\AoUn,`D 
	}{ \derLmuIU `G |- `m`b.[`b]M : \AoU{n} | `D }
~}
 \end{array} \] }
We write $ \derLmuIU `G |- M : A | `D $ if this statement is derivable using these rules.
 \end{definition}
We will normally not distinguish between the two variants of $(\unE)$.

Notice that the traditional $(\arrE)$ of Definition~\ref{strict type assignment definition} is obtained by taking $n = 1$.
Moreover, all $A_i$ can be intersection types, so each can be $\Top$; this is why that rule is not formulated using $ \derLmuIU `G |- N : \AoI{n} | `D $.
If $x{:}\BoUm \ele `G$, then we can only derive $ \derLmuIU `G |- x : \BoU{m} | `D $, \emph{i.e.} we have no way of eliminating a union assigned to a term variable.
Moreover, we have no traditional rules $(\unI)$ and $(\unE)$ on terms, which would be formulated (as in \cite{Barbanera-Dezani-Liguoro-IaC'95}), via
 \[ \begin{array}{rl@{\tquad}rl}
(\unI) : &
\Inf	{\derLmu `G |- M : A | `D }
	{\derLmu `G |- M : A\union B | `D }
&
(\unE) : &
\Inf	{ \derLmu `G |- N : A\un B | `D 
	 \quad
	 \derLmu `G,x{:}A |- M : C | `D 
	 \quad
	 \derLmu `G,x{:}B |- M : C | `D 
	}{\derLmu `G |- M[N/x] : C | `D }
 \end{array} \]
These create the subject-reduction problem dealt with in that paper by limiting to parallel reduction.

Notice that both the strict system for the {\LC} and the system for $`l`m$ are true subsystems; the first by not allowing union types, or alternative conclusions, the second by limiting to Curry types.

 \begin{lemma}[Generation lemma] \label{generation}
 \begin{itemize}
 \firstitem If $\derLmuIU `G |- x : A | `D $, then there exists $x{:}B \ele `G$ such that $B \seq A$. 
 \item If $\derLmuIU `G |- `lx.M : A | `D $, then there exists $B_i, \Cotn$ such that $A = \int{n}(B_i \arr C_i)$, and, for all $\iotn$, $ \derLmuIU `G,x{:}B_i |- M : C_i | `D $. 
 \item If $\derLmuIU `G |- MN : A | `D $, then $A = \AoUn$, and for every $\iotn$ there exists $B_i \ele \Types$ such that $\derLmuIU `G |- M : \un{n}(B_i\arr A_i) | `D $ and $\derLmuIU `G |- N : B_i | `D $.
 \item If $\derLmuIU `G |- `m`a.[`b]M : A | `D $, then there are $\Aotn$ such that $A = \AoIn$, and, for every $\iotn$, there are $m_i,m'_i$ with $m'_i \seq m_i$ and $B^i_j~(\forall j \ele \underline{m_i})$ such that $\derLmuIU `G |- M : \un{m'_i}B^i_k | `b{:}\un{m_i}B^i_j,`a{:}A_i,`D $. 
 \end{itemize}
 \end{lemma}

 \begin{proof}
By easy induction.
\QED
 \end{proof}
 
The system $\TurnLmuIU$ does not have \emph{choice}, \emph{i.e.}~we cannot show that, if $\derLmuIU `G |- M : A\union B | `D $, then either $\derLmuIU `G |- M : A | `D $ of $\derLmuIU `G |- M : B | `D $ as would hold in an intuitionistic system. Take:
 \[ \begin{array}{c}
\Inf	[\unE]
	{\Inf	[\arrR]
		{\Inf	[\unE]
			{\Inf	[\intE]
				{\derLmuIU x{:}A |- x : A | `b{:}B,`d{:}A\union (A\arr B) }
				}{\derLmuIU x{:}A |- `m`b.[`d]x : B | `d{:}A\union (A\arr B) }
			}{\derLmuIU {} |- `lx.`m`b.[`d]x : A\arr B | `d{:}A\union (A\arr B) }
	}{\derLmuIU {} |- `m`d . [`d] (`lx.`m`b . [`d]x ) : A\union (A\arr B) | {} }
 \end{array} \]
Notice that we cannot derive $\derLmuIU {} |- `m`d . [`d] (`lx.`m`b . [`d]x ) : A | {} $, nor $\derLmuIU {} |- `m`d . [`d] (`lx.`m`b . [`d]x ) : A\arr B | {} $, since the two occurrences of $[`d]$ need to be typed differently, but with related types.
This is comparable to both $A$ and $A\arrow B$ to be needed as assumption for $x$ to type $`lx.xx$.

We can show that a general $(\intE)$ (for all terms) is admissible.

 \begin{lemma} \label{inte derivable}
If $\derLmuIU `G |- M : \AoI{n} | `D $, then $\derLmuIU `G |- M : A_i | `D $, for all $\iotn$.
 \end{lemma}

 \begin{proof}
Easy.
\QED
 \end{proof}

The following result is standard.

 \begin{lemma} [Thinning \& Weakening] \label{thinning} \label{weakening} \label{closed for seq}
 \begin{enumerate}
 \firstitem Let $\derLmuIU `G |- M : A | `D $; take $`G' = \Set{x{:}B \ele `G \mid x \ele \fv(M)}$ and $`D' = \Set{`a{:}B \ele `D \mid `a \ele \fv(M)}$, then $\derLmuIU `G' |- M : A | `D' $.
 \item Let $\derLmuIU `G |- M : A | `D $, and $ `G' \seq `G $ and $`D \seq `D' $, then $\derLmuIU `G' |- M : A | `D' $.
 \end{enumerate}
 \end{lemma}
 
 \begin{proof}
By easy induction.
\QED
 \end{proof}

As a consequence, the following rules are admissible:

 \[ \begin{array}{rl@{\dquad}rl}
(\Thin) : & 
\Inf	{\derLmuIU `G |- M : A | `D }
	{\derLmuIU \Set{x{:}B \ele `G \mid x \ele \fv(M)} |- M : A | \Set{`a{:}B \ele `D \mid `a \ele \fv(M)} }
\\[5mm]
(\Weak) : & 
\Inf	[`G' \seq `G,`D \seq `D' ]
	{\derLmuIU `G |- M : A | `D }
	{\derLmuIU `G' |- M : A | `D' }
 \end{array} \]

 \section{Subject reduction and expansion}
We will now show our main results, by showing that our notion of type assignment is sound and complete.
We start by showing two variants of the substitution lemma.


 \begin {lemma} [Term substitution lemma] \label {term substitution lemma}
Let $A$ be strict; $ \derLmuIU `G |- M[N/x] : A | `D $ if and only if there exists $C \ele \Types$ such that $ \derLmuIU `G,x{:}C |- M : A | `D $ and $ \derLmuIU `G |- N : C | `D $.
 \end {lemma}


 \begin{proof}
By induction on $M$.

 \begin {description} \itemsep 4pt

 \item [$ M\same x :$]

 \begin {description} 

 \shiftitem [$ \Then \, :$~]
If $\derLmuIU `G |- x[N/x] : A $, then $ \derLmuIU `G,x{:}A |- x : A $ and $ \derLmuIU `G |- N : A $.

 \item[$ \If \, :$~]
If $ \derLmuIU `G |- x : A | `D $, then there exists $\Aotn$ such that $A = A_k$ from some $k \ele \n$, and $`G = `G',x{:}\AoIn$, so $ \derLmuIU `G',x{:}\AoI{n} |- x : A_k | `D $.
From $\derLmuIU `G |- N : \AoI{n} | `D $ and Lemma~\ref {inte derivable}, we have $ \derLmuIU `G |- N : A | `D $, so $ \derLmuIU `G |- x[N/x] : A | `D $.

 \end {description}

 \item [$ M\same y\not= x :$]

 \begin {description}

 \shiftitem [$ \Then \, :$]
By Lemma \ref{thinning}, since $y[N/x] \same y$, and $x \notele \FV{y}$.

 \item [$ \If \, :$] $\derLmuIU `G |- y[N/x] : A | `D \Then \derLmuIU `G |- y : A | `D $.
Take $ C = \Top$; by Lemma~\ref{closed for seq}, $\derLmuIU `G,x{:}{\Top} |- y : A | `D $.

 \end {description}

\Comment{
 \item [$M\same `l y.M' :$]
By induction.

Because of Lemma~\ref{inte derivable}, we can assume $A$ is strict.

$ \kern3mm \begin{array}[t]{ll}
 \Exists C \Pr[\derLmuIU `G,x{:}C |- `ly.M : A | `D \And \derLmuIU `G |- N : C | `D ] & \Iff (\arrI)
\\
 \Exists C, `m, `j \Pr[\derLmuIU `G,x{:}C,y{:}`m |- M : `j | `D \And A = `m\arr `j \And \derLmuIU `G |- N : C | `D ] & \Iff (\IH)
\\
 \Exists `m, `j \Pr[\derLmuIU `G,y{:}`m |- M[N/x] : `j | `D \And A = `m\arr `j ] & \Iff (\arrI)
\\
 \derLmuIU `G |- `ly.(M[N/x]) : A | `D & \Iff
\\
 \derLmuIU `G |- (`ly.M)[N/x] : A | `D
 \end{array} $
}

 \item [$ M' = M_1M_2 :$] 
Let $A = \un{r}A_j$, with $r \geq 1$.
Notice that $(M_1M_2)[N/x] = M_1[N/x]M_2[N/x]$.

\begin{description} 
 
\item [$\Then \, : $~] 
Then, by Lemma~\ref{generation}, there are $D_j \ele \Types~(\forall j \ele \r)$ such that $\derLmuIU `G |- M_1\,[ N/x ] : \un{r}(D_j\arr A_j) | `D $ and $\derLmuIU `G |- M_2\,[ N/x ] : D_j | `D $, for all $j \ele \r$.
Then by induction, there are $C_1, C^1_2, \ldots, C^r_2 $ such that: 

 \begin{itemize}

 \item $ \derLmuIU `G,x{:}C_1 |- M_1 : \un{r}(D_j\arr A_j) | `D $ and $\derLmuIU `G |- N : C_1 | `D $, as well as

 \item $ \derLmuIU `G,x{:}C^j_2 |- M_2 : D_j | `D $  and $\derLmuIU `G |- N : C^j_2 | `D $, for all $j \ele \r$.

 \end{itemize}
Take $C = C_1 \int C^1_2 \int \dots \int C^r_2 $;
then by weakening and $(\arrE)$, we get $ \derLmuIU `G,x{:}C |- M_1M_2 : A | `D $; notice that $\derLmuIU `G |- N : C | `D $ by $(\intI)$.

 \item [$\If \, :$~] 
If $ \derLmuIU `G,x{:}C |- M_1M_2 : \un{r}A_j | `D $, then by Lemma~\ref{generation} there exists $D_j \ele \Types~(j \ele \r)$ such that $\derLmuIU `G,x{:}C |- M_1 : \un{r}(D_j\arr A_j) | `D $ and $\derLmuIU `G,x{:}C |- M_2 : D_j | `D $, for $j \ele \r$.
Then, by induction, $\derLmuIU `G |- M_1 [ N/x ] : \un{r}(D_j\arr A_j) | `D $ and $\derLmuIU `G |- M_2 [ N/x ] : D_j | `D $ for all $j \ele \r$; the result follows by $(\arrE)$.

 \end{description}

 \item [$ M\same `l y.M' $; $ M\same {`m`a.[`b]}M' :$]
By induction.
\Comment{
$ \kern-14mm \begin{array}[t]{ll}
 \kern14mm 
\Exists C \Pr[\derLmuIU `G,x{:}C |- [`a]M' : A | `D \And \derLmuIU `G |- N : C | `D ] & \Iff (\unE)
\\
\multicolumn{2}{l}{%
A = \bot \And \Exists C,\Aotn \Pr[\derLmuIU `G,x{:}C |- M' : A_i | `a{:}`D \And `a{:}\AoUn \ele `D \And \derLmuIU `G |- N : C | `D ]
} \\
 & \Iff (\IH)
\\
A = \bot \And \Exists \Aotn \Pr[\derLmuIU `G |- M'[N/x] : A_i | `a{:}`D \And `a{:}\AoUn \ele `D & \Iff (\unE)
\\
 \derLmuIU `G |- ([`a]M')[N/x] : A | `D
 \end{array}$

 \item [$(M\same `m`a.M')$]
Wlog, assume $A$ is non an intersection.

$ 
\begin{array}[t]{ll}
\Exists C \Pr[\derLmuIU `G,x{:}C |- `m`a.M' : A | `D \And \derLmuIU `G |- N : C | `D ] & \Iff (`m)
\\
\Exists C \Pr[\derLmuIU `G,x{:}C |- M' : {\bot} | `a{:}A,`D \And \derLmuIU `G |- N : C | `D ] & \Iff (\IH)
\\
\derLmuIU `G,x{:}C |- M'[N/x] : {\bot} | `a{:}A,`D & \Iff (`m)
\\
 \derLmuIU `G |- (`m`a.M')[N/x] : A | `D
 \end{array}$

{\Redcol Here the type derived can be an intersection of union types; do I want that?}
}

 \end {description}
 \end{proof}

Because of Lemma~\ref{inte derivable}, we can extend the above results also to the case that $A$ is an intersection type; notice that this is implicitly used in the third case, where $D_j$ can be an intersection type.

Dually, we have:


 \begin{lemma} [Structural substitution lemma] \label{structural substitution lemma}
$\derLmuIU `G |- M' [ N{`.}`g/`a ] : C | `g{:}\BoUn,`D $ if and only if there are 
%
$\Aotn$ such that for every $A_i$ there exists a $B_i$ such that, for all $\iotn$, $\derLmuIU `G |- N : A_i | `D $, and $ \derLmuIU `G |- M' : C | `a{:}\un{n}(A_i\arr B_i),`D $.
 \end{lemma}

 \begin{proof}
We only show the interesting cases.

 \begin{description} \itemsep3\point

\item [$ M' = x :$] 
Then $x [ N{`.}`g/`a ] = x $; as above the result follows, in either direction, by thinning and weakening.

\item [$ M' = `lx.M :$] 
By induction. %

 \item [$ M' = M_1M_2 :$] 
Then $ M_1M_2 [ N{`.}`g/`a ] = M_1 \, [ N{`.}`g/`a ]~M_2 \, [ N{`.}`g/`a ] $; assume $C$ is strict.

\begin{description} 
 
\item [$\Then \, :$~] 
Let $C = \un{r}C_j$, with $r \geq 1$.
Then, by Lemma~\ref{generation}, there exists $D_j \ele \Types~(\forall j \ele \r)$ such that $\derLmuIU `G |- M_1\,[ N{`.}`g/`a ] : \un{r}(D_j\arr C_j) | `g{:}\BoUm,`D $ and $\derLmuIU `G |- M_2\,[ N{`.}`g/`a ] : D_j | `g{:}\BoUm,`D $, for $j \ele \r$.
Then by induction, there are $A_i ~ (\forall i \ele \k)$ and $A'_i ~(\forall i \ele \l)$ with $k + l = m$ such that 

 \begin{itemize}

 \item 
$ \derLmuIU `G |- M_1 : \un{r}(D_j\arr C_j) | `a{:}\un{k}(A_i\arr B_i),`D $ and, for all $i \ele \k$, $\derLmuIU `G |- N : D_i | `D $, as well as

 \item
$ \derLmuIU `G |- M_2 : D_j | `a{:}\un{l}(A'_i\arr B_i),`D $ for all $j \ele \r$, and, for all $i \ele \l$, $\derLmuIU `G |- N : A'_i | `D $.

 \end{itemize}
Then by weakening and $(\arrE)$, we get $ \derLmuIU `G |- M_1M_2 : \un{r}C_j | `a{:}\un{k}(A_i\arr B_i) \Union \un{l}(A'_i\arr B_i),`D $; notice that $\derLmuIU `G |- N : F | `D $ for all $F \ele \Set{A_i ~ (\forall i \ele \k),~A'_i ~(\forall i \ele \l)}$.

 \item [$\If \, :$~] 
If $ \derLmuIU `G |- M_1M_2 : \un{r}C_j | `a{:}\un{n}(A_i\arr B_i),`D $, then $A = \AoUn$, and there are $D_j \ele \Types~(\forall j \ele \r)$ such that $\derLmuIU `G |- M_1 : \un{r}(D_j\arr C_j) | `a{:}\un{n}(A_i\arr B_i),`D $ and $\derLmuIU `G |- M_2 : D_j | `a{:}\un{n}(A_i\arr B_i),`D $, for $j \ele \r$.
Then, by induction, $\derLmuIU `G |- M_1 [ N{`.}`g/`a ] : \un{r}(D_j\arr C_j) | `g{:}\BoUm,`D $ and $\derLmuIU `G |- M_2 [ N{`.}`g/`a ] : D_j | `g{:}\BoUm,`D $ for all $j \ele \r$; the result follows by $(\arrE)$.
 \end{description}

 \item [{$ M' = `m`b.[`a]M :$}] 
 \begin{description}
 \item[$\Then \, :$ ]
Notice that $ `m`b.[`a]M [ N{`.}`g/`a ] = `m`b.[`g](M\,[ N{`.}`g/`a ] N) $ by definition.
From $\derLmuIU `G |- `m`b.[`g](M\,[ N{`.}`g/`a ] N) : C | `g{:}\BoUn,`D $, by Lemma~\ref{generation}, there are $r < n$ and $E_l,D_l ~ (\forall l \ele \r)$ such that, without loss of generality, $\un{r}(E_l\arr D_l) \Union \un{n{-}r}B_i = \un{n}B_i $, and the derivation is shaped like (notice that we can assume $`g,`b \not\in \fv(N)$):
 \[ \begin{array}{c}
\Inf	[\unE]
	{\Inf	[\arrE]
		{\Inf	[\Weak]
			{\InfBox{ \derLmuIU `G |- M\,[ N{`.}`g/`a ] : \un{r}(E_l\arr D_l) | `g{:}\un{n{-}r}B_i,`b{:}C,`D }
			}{ \derLmuIU `G |- M\,[ N{`.}`g/`a ] : \un{r}(E_l\arr D_l) | `g{:}\BoUn,`b{:}C,`D }
		 \quad
		 \InfBox{\derLmuIU `G |- N : E_l | `D } ~ (\forall l \ele \r) 
		}{ \derLmuIU `G |- M\,[ N{`.}`g/`a ] N : \un{r}D_l | `g{:}\BoUn,`b{:}C,`D }
	}{ \derLmuIU `G |- `m`b.[`g](M\,[ N{`.}`g/`a ] N) : C | `g{:}\BoUn,`D }
 \end{array} \]
Then, by induction, there exist $A_i ~ (\forall i \ele n{-}r)$ such that $ \derLmuIU `G |- M : B_n | `a{:}\un{n{-}r}(A_i\arr B_i),`b{:}C,`D $ and, for all $i \ele n{-}r$, $\derLmuIU `G |- N : A_i | `D $, and we can construct:
 \[ \begin{array}{c}
\Inf	[\unE]
	{\Inf	[\Weak]
		{\InfBox{ \derLmuIU `G |- M : \un{r}(E_l\arr D_l) | `a{:}\un{n{-}r}(A_i\arr B_i),`b{:}C,`D }
		}{ \derLmuIU `G |- M : \un{r}(E_l\arr D_l) | `a{:} \un{r}(E_l\arr D_l) \Union \un{n{-}r}(A_i\arr B_i),`b{:}C,`D }
	}{ \derLmuIU `G |- `m`b.[`a]M : C | `a{:} \un{r}(E_l\arr D_l) \Union \un{n{-}r}(A_i\arr B_i),`D }
 \end{array} \]
Notice that also $ \derLmuIU `G |- N : D | `D' $ for every $D \ele \Set{ E_1, \ldots, E_r, A_1, \ldots, A_{n{-}r} }$.

 \item[$\If\, :$ ]
If $ \derLmuIU `G |- `m`b.[`a]M : C | `a{:} \un{n}(A_i\arr B_i),`D $ and $ \derLmuIU `G |- N : A_i | `D' $ for every $\iotn$, then, by Lemma \ref{generation}, this derivation is constructed as follows:
 \[ \begin{array}{c}
\Inf	[\unE]
	{ \InfBox{ \derLmuIU `G |- M : \un{r}(A_i\arr B_l) | `a{:} \un{n}(A_i\arr B_i),`b{:}C,`D }
	}
	{ \derLmuIU `G |- `m`b.[`a]M : C | `a{:}\un{n}(A_i\arr B_i),`D }
 \end{array} \]
for some $r \seq n$. 
Then, by induction, $ \derLmuIU `G |- M [ N{`.}`g/`a ] : \un{r}(A_i\arr B_i) | `g{:}\BoUn,`b{:}C,`D $, and we can construct:
 \[ \begin{array}[b]{cc}
\dquad
\Inf	[\unE]
	{\Inf	[\arrE]
		{\InfBox{ \derLmuIU `G |- M [ N{`.}`g/`a ] : \un{r}(A_i\arr B_i) | `g{:}\BoUn,`b{:}C,`D }
		 \dquad
		 \InfBox{\derLmuIU `G |- N : A_i | `D } ~ (\forall i \ele \r)
		}{ \derLmuIU `G |- (M [ N{`.}`g/`a ])N : \un{r}B_i | `g{:}\BoUn,`b{:}C,`D }
	}{ \derLmuIU `G |- `m`b.[`g]M [ N{`.}`g/`a ] : C | `g{:}\BoUn,`D }
&
\\[-10\point] \QED
 \end{array} \]
 \end{description}
%
%

 \end{description}
 \end{proof}

Using these two lemmas, we can prove the two main results of this paper:

 \begin{theorem}[Subject expansion]
If $M \redlmu N$, and $ \derLmuIU `G |- N : A | `D $ ($A$ strict), then $ \derLmuIU `G |- M : A | `D $. 
 \end{theorem}
 \begin{proof}
By induction on the definition of reduction, where we focus on the reduction rules.

 \begin{description} \itemsep 2\point

 \item[{$ (`l x . M ) N \red M [ N /x ] : $}]
If $ \derLmuIU `G,x{:}B |- M[N/x] : A | `D $, then by Lemma~\ref{term substitution lemma} there exists a $B \ele \Types$ such that $ \derLmu `G,x{:}B |- M : A | `D $ and $ \derLmuIU `G |- N : B | `D $; then, by applying rule $(\arrI)$ to the first result we get $ \derLmuIU `G |- `lx.M : B\arr A | `D $ and then by $(\arrE)$ we get $ \derLmuIU `G |- (`lx.M)N : A | `D $.

 \item[{$ (`m `a . [`a]M ) N \red `m `g . [`g]M [ N{`.}`g/`a ]N : $}]
If $ \derLmuIU `G |- `m `g . [`g]M [ N{`.}`g/`a ]N : A | `D $, then $A = \AoUn$, and by Lemma~\ref{generation}, (wlog) 
there is $m \seq n$ such that $ \derLmuIU `G |- M [ N{`.}`g/`a ]N : \AoU{m} | `g{:}\AoUn,`D $, and there are $\Botm$ such that, $ \derLmuIU `G |- M [ N{`.}`g/`a ] : \un{m}(B_j \arr A_j) | `g{:}\AoUn,`D $ and for all $\jotm$, $\derLmuIU `G |- N : B_j | `D $.
Then, by Lemma~\ref{structural substitution lemma}, there are $\Cotn$ such that for all $\iotn$, $\derLmuIU `G |- N : C_i | `D $, and $ \derLmuIU `G |- M : \un{m}(B_j\arr A_j) | `a{:}\un{n}(C_i\arr A_i),`D $; (wlog) by weakening, we can assume $\BoUm \seq \CoUn$.
Then, by rule $(\unE)$, $ \derLmuIU `G |- `m`a.[`a]M : \un{n}(C_i\arr A_i) | `D $, and $ \derLmuIU `G |- (`m`a.[`a]M)N : \AoU{n} | `D $ then follows by rule $(\arrE)$.
\Comment{
 \item[{$ (`m `a . [`b]M ) N \red `m `g . [`b]M [ N{`.}`g/`a ] : $}] Similar to the previous part - but easier.
}

 \item[{$ (`m `a . [`b]M ) N \red `m `g . [`b]M [ N{`.}`g/`a ] : $}]
If $ \derLmuIU `G |- (`m`a.[`b]M)N : A | `D $, then $A = \AoUn$, and by $(\arrE)$ there are $\Cotn$ such that $ \derLmuIU `G |- `m`a.[`b]M : \un{n}(C_i\arr A_i) | `D $, and $ \derLmuIU `G |- N : C_i | `D $ for all $\iotn$; then by Lemma~\ref{structural substitution lemma}, $ \derLmuIU `G |- `m `g . [`b]M [ N{`.}`g/`a ] : A | `D $.
\Comment{
 \item[{$ (`m `a . M ) N \red `m `g . M [ N{`.}`g/`a ] : $}]
If $ \derLmuIU `G |- (`m`a.M)N : A | `D $, then $A = \AoUn$, and by $(\arrE)$ there are $\Cotn$ such that $ \derLmuIU `G |- `m`a.M : \un{n}(C_i\arr A_i) | `D $, and $ \derLmuIU `G |- N : C_i | `D $ for all $\iotn$; then by Lemma~\ref{structural substitution lemma}, $ \derLmuIU `G |- `m `g . M [ N{`.}`g/`a ] : A | `D $.

 \item[{$ `m`a.[`b](`m`g.[`d]M) \red `m`a.[`d](M[`b/`g]) $}]
If $ \derLmuIU `G |- M[`b/`g] : A | `D $, then there are $\Botn$ such that $`b{:}\BoUn,`D' = `D$. Since $M$ can contain $`b$ as well, this means that there are $C_j\,(\forall \jotm),D_i\,(\forall i \ele \k)$ such that $\CoUm \Union \un{\k}D_i = \BoUn$, and $ \derLmuIU `G |- M : A | `g{:}\CoUm,`b{:}\un{\k}D_i,`D' $.

Then $ \derLmuIU `G |- `m`g.M : \CoU{m} | `b{:}\un{\k}D_i,`D' $, and $ \derLmuIU `G |- [`b]`m`g.M : {\bot} | `b{:}\BoUn,`D' $.
}

 \item[{$ `m`a.[`b]`m`g.[`d]M \red `m`a.[`d](M[`b/`g]) : $}]
If $ \derLmuIU `G |- `m`a.[`d](M[`b/`g]) : A | `D $, then by rule $(\unE)$, there exist $`d{:}\DoUn \ele `D$ and $m \seq n$ such that $ \derLmuIU `G |- M[`b/`g] : \DoU{m} | `a{:}A,`D $.
Let $`D = `d{:}\DoUn,`b{:}\BoUk,`D'$.
Since $M$ can contain $`b$ as well, this means that there are $C_j\,(\forall j \ele \k),E_i\,(\forall i \ele \l)$ with $\CoUm \Union \un{\k}E_i = \BoUk$, and we can construct:
 \[ \begin{array}{c}
\Inf	[\unE]
	{\Inf	[\unE]
		{\InfBox{ \derLmuIU `G |- M : D_k | `g{:}\CoUm,`d{:}\DoUn,`b{:}\un{\k}E_i,`a{:}A,`D' }
		}{ \derLmuIU `G |- `m`g.[`d]M : \CoU{m} | `d{:}\DoUn,`b{:}\un{\k}E_i,`a{:}A,`D' }
	}{ \derLmuIU `G |- `m`a.[`b]`m`g.[`d]M : A | `D }
 \end{array} \]
which shows the result.
\QED

\Comment{
 \item[{$ [`b](`m`g.M) \red M[`b/`g] $}]
If $ \derLmuIU `G |- M[`b/`g] : A | `D $, then there are $\Botn$ such that $`b{:}\BoUn,`D' = `D$. Since $M$ can contain $`b$ as well, this means that there are $C_j\,(\forall \jotm),D_i\,(\forall i \ele \k)$ such that $\CoUm \Union \un{\k}D_i = \BoUn$, and $ \derLmuIU `G |- M : A | `g{:}\CoUm,`b{:}\un{\k}D_i,`D' $.

{\Redcol Then in order to apply rule $(`m)$, $A$ has to be $\bot$.}

Then $ \derLmuIU `G |- `m`g.M : \CoU{m} | `b{:}\un{\k}D_i,`D' $, and $ \derLmuIU `G |- [`b]`m`g.M : {\bot} | `b{:}\BoUn,`D' $.

 \item[{$ `m `a . [`a] M \red M $, if $`a$ does not occur in $M$}]
Let $A = \AoUn$.
If $ \derLmuIU `G |- `m`a .[`a]M : A | `D $, then, via rules $(`m)$ and $(\unE)$ also $ \derLmuIU `G |- M : A | `a{:}A,`D $.
Since $`a \not\in \fv(M)$, by Thinning, also $ \derLmuIU `G |- M : A | `D $.
}
 
 \end{description}
 \end{proof}

 \begin{theorem}[Subject reduction]
If $M \redlmu N$, and $ \derLmu `G |- M : A | `D $, where$A$ is not an intersection, then $ \derLmu `G |- N : A | `D $ 
 \end{theorem}
 \begin{proof}
 \begin{description} \itemsep 2\point

 \item[{$(`l x . M ) N \red M [ N /x ] $}]
Let $ \derLmu `G |- (`lx.M)N : A | `D $.
Then by Lemma~\ref{generation} there exists $B \ele \Types$ such that $ \derLmu `G |- `lx.M : B\arr A | `D $ and $ \derLmu `G |- N : B | `D $, and also $ \derLmu `G,x{:}B |- M : A | `D $.
Then, by Lemma~\ref{term substitution lemma}, we have $ \derLmu `G |- M[N/x] : A | `D $.

 \item[{$ (`m `a . [`a]M ) N \red `m `g . [`g]M [ N{`.}`g/`a ]N : $}]
If $ \derLmuIU `G |- (`m`a.[`b]M)N : A | `D $, then by Lemma~\ref{generation} there exist $\Aotn$ and $\Cotn$ such that $A = \AoUn$, and $ \derLmuIU `G |- `m`a.[`b]M : \un{n}(C_i\arr A_i) | `D $ and, for all $\iotn$, $ \derLmuIU `G |- N : C_i | `D $; then also $ \derLmuIU `G |- M : B | `a{:}\un{n}(C_i\arr A_i),`D $, with $`b{:}B' \ele `D$ with $B$ and $B'$ union types such that $B \seq B'$.
Then, by Lemma~\ref{structural substitution lemma}, $ \derLmuIU `G |- M[ N{`.}`g/`a ] : B | `g{:}\AoUn,`D $, so, by rule $(\unE)$, $ \derLmuIU `G |- `m`g.[`b]\,M[ N{`.}`g/`a ] : \AoU{n} | `D $.
Then, by $(\arrE)$, $ \derLmuIU `G |- MN : C | `a{:}\un{n}(A_i\arr B_i),`D $.
\Comment{
 \item[{$ (`m `a . [`b]M ) N \red `m `g . [`b]M [ N{`.}`g/`a ] : $}] Similar to the previous part - but easier.
}

 \item[{$ `m`a.[`b]`m`g.[`d]M \red `m`a.[`d](M[`b/`g]) $}]
If $ \derLmuIU `G |- `m`a.[`b]`m`g.[`d]M : A | `D $, the derivation is shaped like:
 \[ \begin{array}{c}
\Inf	[\unE]
	{\Inf	[\unE]
		{\InfBox{ \derLmuIU `G |- M : D_p | `g{:}B_l,`d{:}\DoUn,`b{:}\un{\k}B_i,`a{:}A,`D' } 
		}{ \derLmuIU `G |- `m`g.[`d]M : B_l | `b{:}\un{\k}B_i,`a{:}A,`D' }
	}{ \derLmuIU `G |- `m`a.[`b]`m`g.[`d]M : A | `b{:}\un{\k}B_i,`D' }
 \end{array} \]
with $`D = `b{:}\BoUm,`D'$, for some $\Botm$, with $l \ele \k$, and $p \ele \n$.
It is straightforward to show that then $\derLmuIU `G |- M[`b/`g] : D_p | `b{:}\un{\k}B_i,`a{:}A,`D' $, and applying rule $(\unE)$ to this derivation gives 
$\derLmuIU `G |- `m`a.[`d](M[`b/`g]) : A | `D $.
\QED

 \end{description}
 \end{proof}

Notice that we cannot show subject reduction for the \emph{erasing} rule.
Assume the derivation for $ `m `a . [`a] M $ with $M$ not a control structure is shaped like
 \[ \begin{array}{c}
\Inf	[`m]
	{\Inf	[\unE]
		{\InfBox{ \derLmuIU `G |- M : A_j | `a{:}\AoUn,`D }
		}{\derLmuIU `G |- [`a]M : {\bot} | `a{:}\AoUn,`D }
	}{ \derLmuIU `G |- `m`a.[`a]M : \AoU{n} | `D }
 \end{array} \]
Since $`a$ does not appear in $M$, by thinning we can derive $ \derLmuIU `G |- M : A_j | `D $, but have no rule to allow us to derive $ \derLmuIU `G |- M : \AoU{n} | `D $ from that.

\section*{Conclusions}
We have seen that the calculus $\lmu$ is sufficiently limited to allow for the definition of a sound and complete notion of type assignment.
This will need to be investigated further, towards the definition of semantics, and characterisation of the termination properties.
Also, we need to look at the ignored reduction rules, and see if it is possible to generalise the system such that also these will be preserved, without sacrificing the main results. 
The approach we use here seems to be promising also for the setting of (restrictions of) $\X$ and $\lmmt$; we will leave this for future work.


{

}

 \end{document}

%% file: Xmacros.tex
\def \mysl#1{\mbox{\rm\textsl{#1}}\kern0\point}
\def \mytiny{\fontsize{6.5}{9}\selectfont}
\def \tinyfont{\fontsize{6pt}{6.6pt}\selectfont}

\makeatletter 

\newif\if@om
\def \omfalse{\let\if@om\iffalse}
\def \omtrue{\let\if@om\iftrue}

\newif\if@lom
\def \lomfalse{\let\if@lom\iffalse}
\def \lomtrue{\let\if@lom\iftrue}

\newskip \adjustablepoint \adjustablepoint .9\point

\usepackage{pstricks}
\usepackage{pst-node}
\usepackage{pdftricks}
\usepackage{graphicx} 
\usepackage{qsymbols}
\usepackage{inference}

\Comment{
}

\def \typeout#1{\@latex@warning{#1}}

\def \Exists {\exists\,}

\def \ifnextchar{\@ifnextchar}

\def \qskp {\hspace*{.25\point}}
\def \hskp {\hspace*{.5\point}}
\def \skp {\hspace*{1\point}}
\def \skipper {\hspace*{1.5\point}}
\def \Bogskipleft	{~\kern-1\point~}
\def \Bogskipright	{\ \kern-1\point\ }
\def \skipleft	{\hspace*{2\point}}
\def \skipright {\hspace*{2\point}}

\def \dquad {\quad \quad}
\def \tquad {\quad \quad \quad}
\def \qquad {\quad \quad \quad \quad}

\def \BoUndnccurve[#1]#2#3{\psset{linecolor=red}%
\nccurve[#1]{#2}{#3}\psset{linecolor=ruleyellow}}

\def \It#1{\mbox{\rmfamily \it \Itcol #1}}

\def \mynewrgbcolor#1#2{%
\ifmycolour\newrgbcolor{#1}{#2}\else \newrgbcolor{#1}{0 0 0}\fi}

\mynewrgbcolor{ltermcol}{.7 0 .7}
\mynewrgbcolor{termcol}{1 0 0}
\mynewrgbcolor{typecol}{0 .50 0}
\mynewrgbcolor{rulegreencol}{.1 1 .1}
\mynewrgbcolor{DarkGreencol}{0.1 .5 0.1}
\mynewrgbcolor{Greencol}{0.1 .5 0.1}
\mynewrgbcolor{Yellowcol}{.9 .9 0}
\mynewrgbcolor{Itcol}{.9 .5 0}
\mynewrgbcolor{Blackcol}{0 0 0}
\mynewrgbcolor{Bluecol}{0.35 0.35 1}
\mynewrgbcolor{Orangecol}{1 0.35 0.25}
\mynewrgbcolor{Purplecol}{.75 0 .75}
\mynewrgbcolor{Whitecol}{1 1 1}
\mynewrgbcolor{Greencol}{0 .25 0}
\mynewrgbcolor{yellowcol}{1 1 0}
\mynewrgbcolor{orangecol}{1 0.75 0.5}
\mynewrgbcolor{bluecol}{0 0 1}
\mynewrgbcolor{Redcol}{1 0.05 0.05}
\mynewrgbcolor{lightgreencol}{0.75 1 0.75}
\mynewrgbcolor{lightyellowcol}{1 1 .75}
\mynewrgbcolor{lightbluecol}{.75 .75 1}
\mynewrgbcolor{lightredcol}{1 .5 .5}
\mynewrgbcolor{purplecol}{.9 0.1 .9}
\mynewrgbcolor{whitecol}{1 1 1}
\mynewrgbcolor{blackcol}{0 0 0}
\mynewrgbcolor{neutralcol}{0 0 0}

\newrgbcolor{Redcol}{1 0.05 0.05}

\mynewrgbcolor{ruleyellow}{1 1 .5}
\mynewrgbcolor{rulegreen}{.1 1 .1}

\mynewrgbcolor{foreground}{0 0 0}
\mynewrgbcolor{background}{1 1 1}

\def \ltermcolour{\ifmycolour\ltermcol\fi}
\def \termcolour{\ifmycolour\termcol\fi}
\def \Yellow{\ifmycolour\Yellowcol\fi}
\def \typecolour{\ifmycolour\typecol\fi}

\def \Yellow{\ifmycolour\Yellowcol\fi}
\def \Black{\ifmycolour\Blackcol\fi}

\def \Purple{\ifmycolour\Purplecol\fi}

\def \orange{\ifmycolour\orangecol\fi}
\def \blue{\ifmycolour\bluecol\fi}
\def \Red{\ifmycolour\Redcol\fi}

\def \black{\ifmycolour\blackcol\fi}

\def \specialcol{\termcolour}

\def \mysemcolour {\mynewrgbcolor{thiscol}{1 0 1}\thiscol}
\ifmycolour \mynewrgbcolor{semcolour}{.9 0 0}\else \mynewrgbcolor{semcolour}{0 0 0}\fi

\psset{linecolor=ruleyellow,linewidth=0.3\point,arrows=->,nodesep=0.05}
\def \PSFrame(#1,#2){\mbox{\psframe[unit=\unitlength,linewidth=.5\point,framearc=.5
](#1,#2)}}
\def \PSBox(#1,#2){\mbox{\psframe[unit=\unitlength,linewidth=.5\point,framearc=0
](#1,#2)}}

\def \Set#1{\{\,#1\,\}}
\def \Pr[#1]{~[\,#1\,]}
\def \ele {\mathbin{\in}}
\def \notele {\mathbin{\not\nobreak\in}}

\def \subssymb {(\tvar \mapsto C)}
\def \subsone#1{\subssymb\,(#1)}
\def \subs{\@ifnextchar\bgroup{\subsone}{\subssymb}}

\def \unifytwo#1#2{\It {unify}\ #1\ #2}
\def \unify{\@ifnextchar{\bgroup}{\unifytwo}{\It {unify}}}

\def \equ	{\mathbin{\sim}}

\def \beq	{\skipleft{\geq}\skipright}
\def \seqr {{\typecolour \leq}}

\def \seq	{\mathbin{\seqr}}

\def \Union {\mathbin{\cup}}

\def \And	{\mathrel{\&}}

\def \Conj	{\mathord{\wedge}}

\def \longarrow#1{\psset{unit=1\point,linewidth=0.35\point}%
\psline{cc-cc}(0,0)(#1,0)%
\hspace*{#1}%
\pscurve{cc-cc}(0,0)(-1.1,.4)(-2.5,2.2)%
\pscurve{cc-cc}(0,0)(-1.1,-.4)(-2.5,-2.2)}

\def \Rel#1#2{\setbox155=\hbox{\scriptsize $#1$}\setbox156=\hbox{\scriptsize $#2$}%
\,\raise5.5pt\copy155\kern-1\wd155%
\raise3pt\hbox{\longarrow{1\wd155}}
\kern-.5\wd155\kern-.5\wd156\raise-2pt\copy156\kern-.5\wd156\kern.5\wd155\,}

\def \dots	{\mbox	{$\cdots$}}
\def \Skip	{\hspace*{2\point}}

\def \pCsymb{\It{pC}}
\def \pCarg#1{\pCsymb\,({\termcolour #1})}
\def \pC{\@ifnextchar\bgroup{\pCarg}{\pCsymb}}

\def \pCgsymb{\It{pC}_g}
\def \pCgarg#1{\pCgsymb\,({\termcolour #1})}
\def \pCg{\@ifnextchar\bgroup{\pCgarg}{\pCgsymb}}

\def \SetApp(#1){{\cal A}(#1)}

\def \bred	{\mathrel{\bredr}}

\def \bredr {\mbox	{\Red $\rightarrow_{\beta}$}}

\def \dbredr	{\arrow \kern-8\point\bredr}

\def \LC	{\Sc{lc}}

\def \diverges {\,{\raise1.5\point\hbox{\small $\uparrow$}}}

\newcommand{\indexot}[2]{{#1}\skp{\in}\skp{#2}}

\def \iotn {\indexot{i}{\n}}

\def \jotn {\indexot{j}{\n}}

\def \jotm {\indexot{j}{\m}}

\def \Top {\mbox{\typecolour $\top$}}
\def \Bottom {\mbox{\typecolour $\perp$}}
\def \Types	{{\Red \cal T}}

\def \Tproper	{\Types\kern-3\point_{\textit{\scriptsize \Yellow p}}}
\def \Tstrict	{\Types_{\textit{\scriptsize s}}}
\def \TStrict	{\Types_{\textit{\scriptsize s}}}

\def \dom(#1){\textit{dom}\,(#1)}

\def \k	{\underline{k}}
\def \l	{\underline{l}}
\def \m	{\underline{m}}
\def \n	{\underline{n}}

\def \r	{\underline{r}}

\def \arrow	{\mathop{\rightarrow}}
\def \arr {\arrow}

\def \tvar	{\varphi}

\def \intsymb	{{\cap}} 
\def \inter {\mathord{\hskp\type{\intsymb}\hskp}}

\def \intTwo#1(#2){\intsymb_{\ul{#1}}\qskp(#2)}
\def \intOne#1{\@ifnextchar({\intTwo{#1}}{\intsymb_{\ul{#1}}\qskp}}
\def \int{\@ifnextchar\bgroup{\intOne}{\inter}}

\def \unsymb	{{\cup}} 
\def \union {\mathord{\hskp\type{\unsymb}\hskp}}
\def \compunion	{\setbox172=\hbox{$\cup$}%
\copy172\kern-.75\wd172^c\kern.25\wd172}

\def \unTwo#1(#2){\unsymb\kern-.5\point_{\ul{#1}}\qskp({#2})}
\def \unOne#1{\@ifnextchar({\unTwo{#1}}{\unsymb\kern-.5\point_{\ul{#1}}\qskp}}
\def \un{\@ifnextchar\bgroup{\unOne}{\union}}

\def \BuildInt({\@ifnextchar{n}{\BuildIntn(}{\@ifnextchar{m}{\BuildIntm(}{\BuildIntAux(}}}
\def \BuildIntAux(#1)#2{\int{#1}{#2}_l}
\def \BuildIntn(#1)#2{\int{n}{#2}_i}
\def \BuildIntm(#1)#2{\int{m}{#2}_j}
\def \BuildUn({\@ifnextchar{n}{\BuildUnn(}{\@ifnextchar{m}{\BuildUnm(}{\BuildUnA(}}}
\def \BuildUnA(#1)#2{\un{#1}{#2}_l}
\def \BuildUnn(#1)#2{\un{n}{#2}_i}
\def \BuildUnm(#1)#2{\un{m}{#2}_j}

\def \AoI#1{\BuildInt(#1){\type A }}
\def \AoU#1{\BuildUn(#1){\type A }}
\def \Aot#1{\type{A_i} ~(\forall i \ele #1)}

\def \AoIn {\AoI{n}}

\def \AoUn {\AoU{n}}

\def \Aotn {\Aot{n}}

\def \BoU#1{\BuildUn(#1){B}}
\def \Bot#1{B_i~(\forall i \ele #1)}

\def \BoUk {\BoU{k}}
\def \BoUm {\BoU{m}}
\def \BoUn {\BoU{n}}
\def \Botm {\Bot{m}}
\def \Botn {\Bot{n}}

\def \Cot#1{C_i\,(\forall i \ele #1)}

\def \CoU#1{\BuildUn(#1){C}}

\def \CoUm {\CoU{m}}
\def \CoUn {\CoU{n}}

\def \Cotn {\Cot{n}}

\def \DoU#1{\BuildUn(#1){D}}

\def \DoUn {\DoU{n}}

\def \DoUn {\DoU{n}}

\def \G_#1{`G\kern-1.5pt_{#1}}
\def \Ga_#1{`G'\kern-1.5pt_{#1}}

\def \intI	{{\Red \intsymb \mbox{\sl I}}}
\def \intE	{{\Red \intsymb \mbox{\sl E}}}
\def \unI	{{\Red \unsymb \mbox{\sl I}}}

\def \unE	{{\Red \unsymb \mbox{\sl E}}}
\def \Ax {{\Red \mbox{\sl Ax}}}
\def \arrI	{{\Red \arrow \mbox{\sl I}}}
\def \arrE	{{\Red \arrow \mbox{\sl E}}}

\def \dcol {\,{::}\,}

\def \arrR	{{\arr}\mbox{\sl R}}

\def \Weak	{\mbox{\sl W}}

\def \Iffs	{\mbox	{$\Leftarrow\kern-6\point\Rightarrow$}}

\def \Iff	{\Skip\,{\Leftarrow\kern-6\point\Rightarrow}\,\Skip}
\def \Thenone#1{\mathrel{{\Rightarrow}\,(#1)}}
\def \Then{\@ifnextchar{\bgroup}{\Thenone}{\mathrel{\Rightarrow}}}
\def \If{\mathrel{\Leftarrow}}
\def \Implyone#1{\mathrel{\Rightarrow\,(#1)}}
\def \Imply{\@ifnextchar{\bgroup}{\Implyone}{\mathrel{\Rightarrow}}}

\def \X {{\Red \mbox{$\mathcal{X}$}}}
\def \Xi {{\Red{\mathcal X}^{\textit{\footnotesize i}}}} 

\def \diagX {{\Red ^d \kern-.5\point \mathcal X}} 
\def \copyX {{\Red ^{\scriptsize \copyright} \kern-.5\point \mathcal X}} 

\setbox134=\hbox{\leavevmode\raise0.5\point\hbox{${\mysemcolour	\lceil}\kern-3\point{\mysemcolour \lceil}\kern-.5pt$}}
\setbox135=\hbox{\raise0.5\point\hbox{$\kern-.5pt{\mysemcolour	\rfloor}\kern-3\point{\mysemcolour	\rfloor}$}}

\def \LeftTop#1{\hspace{2\point}
\setbox66=\hbox{$#1$}%
\ifdim\ht66>5\point \setlength{\unitlength}{.1\point} 
	\psset{unit=.125\ht66,linewidth=0.4\point,linecolor=black}%
\else \setlength{\unitlength}{.2\point} 
	\psset{unit=.8\point,linewidth=0.4\point,linecolor=black}%
\fi
\ifmycolour\psset{linecolor=semcolour}\fi
\setbox67=\hbox{\put(0,1)%
{\psline{c-c}(0,0)(3,0)%
 \psline{cc-cc}(0,0)(0,-8)%
 \psline{cc-cc}(1.5,0)(1.5,-6)%
}}%
\raise7.25\point\box67
\kern3\point}

\def \RightBot{\kern2.75\point
\psset{unit=1\point,linewidth=0.4\point,linecolor=black}%
\ifmycolour\psset{linecolor=semcolour}\fi
\hbox{\put(.25,-2)%
{\psline{c-c}(0,0)(-3,0)%
 \psline{cc-cc}(-1.35,0)(-1.35,6)%
 \psline{cc-cc}(0,0)(0,8)%
}}\kern2.5\point}

\def \Rightbot{\kern3.25\point
\psset{unit=1\point,linewidth=0.4\point,linecolor=black}%
\ifmycolour\psset{linecolor=semcolour}\fi
\hbox{\put(0,-8)%
{\psline{c-c}(0,0)(-3,0)%
\psline{cc-cc}(-1.35,0)(-1.35,5)%
\psline{cc-cc}(0,0)(0,5)%
}}\kern1\point}

\def \OldSem[#1]{{\semcolour [\kern-1.75\point[}\kern.6\point \lterm{#1} \kern.6\point{\semcolour ]\kern-1.75\point]} }

\def \semEl#1{\LeftTop{#1}{#1}\RightBot}
\def \semThree#1#2#3{\setbox136=\hbox{{\scriptsize \Yellow $#2$}}\setbox137=\hbox{#3}%
	(#1)_{\copy136}\kern-\wd136%
	\raise5\point\box137}
\def \SemThree#1#2#3{%
	\setbox136=\hbox{{\scriptsize \termcolour ${#2}$}}%
	\setbox137=\hbox{#3}%
	\semEl{#1}_{\kern-1\point\termcolour \copy136}
	\ifdim\wd136>1pt\kern-\wd136\raise5\point\copy137
	\ifdim\wd136>\wd137 \kern-\wd137\kern\wd136\fi
	\fi}
\def \SemTwo#1#2{\@ifnextchar{\bgroup}{\SemThree{#1}{#2}}{\semEl{#1}\kern-1\point_{#2}}}
\def \Sem#1{\@ifnextchar{\bgroup}{\SemTwo{#1}}{\semEl{#1}}}

\def \ftn{\mbox{\Red \textsc{\scriptsize n}}}
\def \ftv{\mbox{\Red \textsc{\scriptsize v}}}

\def \SemN#1{\semEl{#1}\kern-1\point_{\ftn}}
\def \SemV#1{\semEl{#1}\kern-1\point_{\ftv}}

\def \SemN#1{\semEl{#1}\kern-1\point_{\ftn}}
\def \SemV#1{\semEl{#1}\kern-1\point_{\ftv}}

\def \Neg#1{\Overline{#1}} 
\def \NegBr(#1){({#1)^o}}
\def \NegNoBr#1{{{#1}^o}}
\def \Neg{\@ifnextchar({\NegBr}{\NegNoBr}}

\def \SemGN#1{\langle{#1}\rangle\kern-1\point_{\ftn}^{l}}
\def \SemDN#1{\langle{#1}\rangle\kern-1\point_{\ftn}^{r}}
\def \SemGV#1{\langle{#1}\rangle\kern-1\point_{\ftv}^{l}}
\def \SemDV#1{\langle{#1}\rangle\kern-1\point_{\ftv}^{r}}

\def \SemLx[#1] #2 {\SemThree{#1}{#2}{\mbox{\mysemcolour	\tiny $`l$\sc x}}}

\def \LmmtSemN#1{\semEl{#1}\kern-1\point_{\ftn}}
\def \LmmtSemV#1{\semEl{#1}\kern-1\point_{\ftv}}

\newcounter{sind}0
\newcounter{tind}0
\def \newsigma{\if@om\omfalse`s\setcounter{sind}0\setcounter{tind}0\else\stepcounter{sind}`s_{\thesind}\fi}
\def \newtau{\ifnum\thetind=0\stepcounter{tind}`t\else`t_{\thetind}\fi}

\def \lefttop {\leavevmode \kern2\adjustablepoint
\psset{unit=.9\adjustablepoint,linewidth=0.4\adjustablepoint,linecolor=Redcol}%
\ifmycolour\psset{linecolor=semcolour}\fi
\hbox{\psline{c-c}(0,9)(3,9)\psline{cc-cc}(1.5,9)(1.5,-2)\psline{cc-cc}(0,9)(0,-2)\psline{c-c}(0,-2)(3,-2)}\kern3.5\adjustablepoint}

\def \righttop{\kern2\adjustablepoint
\psset{unit=.9\adjustablepoint,linewidth=0.4\adjustablepoint,linecolor=Redcol}%
\ifmycolour\psset{linecolor=semcolour}\fi
\hbox{\psline{c-c}(1.5,9)(-1.5,9)\psline{cc-cc}(0,9)(0,2)\psline{cc-cc}(1.5,9)(1.5,2)}\kern1\adjustablepoint}

\def \leftsem {\leavevmode \kern2\point
\psset{unit=.9\point,linewidth=0.4\point,linecolor=black}%
\ifmycolour\psset{linecolor=semcolour}\fi
\hbox{\psline{c-c}(0,9)(3,9)\psline{cc-cc}(1.5,9)(1.5,-2)\psline{cc-cc}(0,9)(0,-2)\psline{c-c}(0,-2)(3,-2)}\kern3.5\point}

\def \rightsem{\kern2\point
\psset{unit=.9\point,linewidth=0.4\point,linecolor=black}%
\ifmycolour\psset{linecolor=semcolour}\fi
\hbox{\psline{c-c}(1.5,9)(-1.5,9)\psline{cc-cc}(0,9)(0,-2)\psline{cc-cc}(1.5,9)(1.5,-2)\psline{c-c}(1.5,-2)(-1.5,-2)}\kern1\point}

\def \Sub #1 #2 := #3{{#1} `<{#2}{\,:=\,}{#3}`> }
\def \VSub #1 #2 := #3{{#1} `<\Vec{{#2}{\,:=\,}{#3}}`> }
\def \exsub #1 := #2 { `<{#1}{\,:=\,}{#2}`> }
\def \imsub #1/#2 {[#1/#2]}

\def \LK{\textsc{lk}}

\setbox139=\hbox{{\Yellow \raise0\point\hbox{\rotatebox{20}{$\downarrow$}}\kern-2.5pt\raise1.72\point\hbox{\rotatebox{-20}{$\downarrow$}}}}
 \def \eqX	{\stackrel{\mbox{\scriptsize $\cal X$}}{=}}
\def \eqXCBN	{\mathbin{\eqX\kern-2pt_{\ftn}}}
\def \eqXCBV	{\mathbin{\eqX\kern-2pt_{\ftv}}}

\def \crXCBN	{\mathbin{\downarrow}\kern-.5\point_{\ftn}}
\def \crXCBV	{\mathbin{\downarrow}\kern-.5\point_{\ftv}}

\def \mutilde	{\tilde{\mu}}
\def \mt	{\mutilde}

\def \lmmt	{{\Red \overline{`l}`m\mutilde}}

\def \lmu	{{\Red `l`m}}
\def \ML	{{\Red \textsc{ml}}}
\def \eqX	{\stackrel{{\cal X}}{ = }}
\def \redX {\mathrel{{\rightarrow}\kern-2\point_{\cal X}}}
\def \rtcredX {\mathrel{{\rightarrow}\kern-1pt^*\kern-3\point_{\cal X}}}
\def \red {\mathrel{\arrow}}

\def \rtcred {\mathrel{\arrow\kern-1\point^*}}
\def \dred {\rightarrow\kern-.75em\rightarrow}

\def \redlmu {\mathrel{\arrow\kern-1\point_{`l`m}}}

 \def \redlazy {\mathrel{\Red \arrow_{\kern-.5pt\mbox{\scriptsize \sc l}}}}
 \def \rtcredlazy {\mathrel{\Red \arrow^{\ast}_{\kern-.5pt\mbox{\scriptsize \sc l}}}}
 \def \reds {\mathrel{\Red \arrow_{\kern-.5\point\mbox{\scriptsize \sc s}}}}
 \def \redx {\mathrel{\Red \arrow_{\kern-.5\point\mbox{\scriptsize {\bf x}}}}}
 \def \redxsub {\mathrel{\Red \arrow_{\kern-.5\point\mbox{\scriptsize {\rm :=}}}}}
 \def \rtcredx {\mathrel{\Red \arrow^{*}_{\kern-.5\point\mbox{\scriptsize {\bf x}}}}}
 \def \rtcredxsub {\mathrel{\Red \arrow^{*}_{\kern-.5\point\mbox{\scriptsize {\rm :=}}}}}
 \def \redxV {\mathrel{\Red \arrow_{\kern-.5\point\mbox{\scriptsize {\bf x}{\sc v}}}}}
 \def \redxl {\mathrel{\Red \arrow_{\kern-.5\point\mbox{\scriptsize {\bf x}{\sc l}}}}}
 \def \redxs {\mathrel{\Red \arrow_{\kern-.5\point\mbox{\scriptsize {\bf x}{\sc s}}}}}
 \def \eqxs {\mathrel{\Red =_{\kern-.5\point\mbox{\scriptsize {\bf x}{\sc s}}}}}
 \def \rtcredxs {\mathrel{\Red \arrow^{*}_{\kern-.5\point\mbox{\scriptsize {\bf x}{\sc s}}}}}
 \def \redxsa {\mathrel{\Red \arrow'_{\kern-.5\point\mbox{\scriptsize {\bf x}{\sc s}}}}}
 \def \tcreds {\mathrel{\Red \arrow^{+}_{\kern-.5\point\mbox{\scriptsize \sc s}}}}
 \def \tcredx {\mathrel{\Red \arrow^{+}_{\kern-.5\point\mbox{\scriptsize {\bf x}}}}}
 \def \tcredxsub {\mathrel{\Red \arrow^{+}_{\kern-.5\point\mbox{\scriptsize {\rm :=}}}}}
 \def \tcredl {\mathrel{\Red \arrow^{+}_{\kern-.5\point\mbox{\scriptsize {\sc l}}}}}
 \def \tcredxl {\mathrel{\Red \arrow^{+}_{\kern-.5\point\mbox{\scriptsize {\bf x}{\sc l}}}}}
 \def \tcredxs {\mathrel{\Red \arrow^{+}_{\kern-.5\point\mbox{\scriptsize {\bf x}{\sc s}}}}}
 \def \rtcreds {\mathrel{\Red \arrow^{\ast}_{\kern-.5\point\mbox{\scriptsize \sc s}}}}
 \def \rtcredl {\mathrel{\Red \arrow^{\ast}_{\kern-.5\point\mbox{\scriptsize {\sc l}}}}}
 \def \rtcredx {\mathrel{\Red \arrow^{\ast}_{\kern-.5\point\mbox{\scriptsize {\bf x}}}}}
 \def \rtcredxl {\mathrel{\Red \arrow^{\ast}_{\kern-.5\point\mbox{\scriptsize {\bf x}{\sc l}}}}}
 \def \rtcredxs {\mathrel{\Red \arrow^{\ast}_{\kern-.5\point\mbox{\scriptsize {\bf x}{\sc s}}}}}
\def \redes {\mathrel{\Red \arrow_{\kern-.5\point\mbox{\scriptsize \sc es}}}}
 \def \dreds {\mathrel{\Red \arrow\kern-8\point\arrow_{\kern-.5\point\mbox{\scriptsize \sc s}}}}

\def \Pair<#1,#2>{`<{#1},{#2}`>}
\def \Group<#1>{\langle{#1}\rangle}

\def \IH {\textit{IH}}

\def \bp(#1){{\typecolour \It{bp}({\termcolour #1})}}
\def \bs(#1){{\typecolour \It{bs}({\termcolour #1})}}
\def \BV#1{\textit{bv}(#1)}
\def \bv(#1){\BV{#1}}

\def \fp(#1){\textit{fp}({\termcolour #1})}
\def \fs(#1){\textit{fs}\skp({\termcolour #1})}
\def \fc(#1){\textit{fc}\skp({\termcolour #1})}
\def \FV#1{\textit{fv}\skp({\termcolour #1})}
\def \fv(#1){\FV{#1}}

\def \Mid
{\hspace*{2.5pt}{\black \mid}\hspace*{2.5pt}}

\def \Langle{\kern1\adjustablepoint
\psset{unit=1\adjustablepoint,linewidth=0.4\adjustablepoint,linecolor=ltermcol}%
\put(0,2.5){\psline{c-c}(0,0)(1.7,4.9)%
\pscurve{c-c}(1.7,4.9)(.7,0)(1.7,-4.9)
\psline{c-c}(0,0)(1.7,-4.9)
}\kern3\adjustablepoint}
\def \Rangle{\psset{unit=1\adjustablepoint,linewidth=0.4\adjustablepoint,linecolor=ltermcol}%
\kern3\adjustablepoint
\put(0,2.5){\psline{c-c}(0,0)(-1.7,4.9)%
\pscurve{c-c}(-1.7,4.9)(-.8,0)(-1.7,-4.9)+
\psline{c-c}(0,0)(-1.7,-4.9)
}\kern\adjustablepoint}
\def \PiCell#1#2{\setbox171=\hbox{$#1$}\setbox172=\hbox{$#2$}
{\ltermcolour
\Langle 
	\ifdim\wd171>24\point\skp{\box171}\skp\else{\box171}\fi
\kern1pt{\mid}\kern1pt
	\ifdim\wd172>24\point\skp{\box172}\skp\else{\box172}\fi
\Rangle 
}}
\def \Cell#1#2{\langle#1\skp{\mid}\skp#2\rangle}
\def \cell<#1|#2>{\Cell{#1}{#2}}

\def \repl #1[#2/#3]{\lterm{#1}{\skp\Black [}\lterm{#2}{\Black /}\lterm{#3}{\Black ]} }
\def \term#1 {{\termcolour #1}}
\def \ass#1:#2 {{\termcolour #1}{\black :}{\typecol #2} }
\def \lass#1:#2 {\lterm {#1}{\black :}{\typecol #2} }

\def \type#1{{\typecolour #1}}
\def \lterm#1{{\ltermcolour #1}}
\def \Rule(#1){{\blue (#1) }}

\def \Stat#1:#2{{\ltermcolour #1}{:}{\typecolour #2}}

\def \lstat#1#2{{\ltermcolour #1}\,{:\,}{\typecolour #2}}

\def \turn	{{\black {\vdash}}}
\def \Turn {\mathrel{\turn}}
\def \TurnSimple {\mathrel{\turn}}
\def \TurnD {\mathrel{\turn\kern-3\point_{\mbox{\tinyfont $\textit{\black D}$}}}}
\def \TurnI {\mathrel{\turn\kern-3\point_{\mbox{\tinyfont $\black \cap$}}}}
\def \TurnL {\mathrel{\turn\kern-.5\point_{`l}}}
\def \TurnLK {\mathrel{\Yellow \turn\kern-2\point_{\ftsc{lk}}}}
\def \TurnLKarr {\mathrel{\Yellow \turn\kern-2\point_{\ftsc{lk}({\rightarrow})}}}
\def \TurnR {\mathrel{\Yellow \turn\kern-2\point_{\ftsc{r}}}}
\def \TurnLN {\mathrel{\Yellow \turn\kern-2\point_{\ftsc{ln}}}}
\def \TurnLV {\mathrel{\Yellow \turn\kern-2\point_{\ftsc{lv}}}}

\def \TurnLmmt {\mathrel{\black \turn_{\sclmmt}}}
\def \TurnLmmtSN {\mathrel{\turn_{\sclmmt}^{\ftsc{sn}}}}
\def \TurnLmmtN {\mathrel{\turn\kern-2\point_{\black \ftsc{n}}}}
\def \TurnLmmtV {\mathrel{\turn\kern-2\point_{\black \ftsc{v}}}}
\def \TurnLmmtmin {\mathrel{\turn\kern-2\point_{\black \ftsc{m}}}}
\def \TurnLmmtmin {\mathrel{\turn}}
\def \TurnLmu {\mathrel{\turn\kern-\point_{`l`m}}}
\def \TurnLmuIU {\mathrel{\,\turn\kern-2\point_{`l`m}^{\,\raise 1\point\hbox{\tinyfont $\cap\cup$}}\,}}
\def \TurnLx {\mathrel{\turn_{`l\mbox{\scriptsize \sf x}}}}
\def \TurnMCL {\mathrel{\Yellow \turn\kern-2\point_{\ftsc{mcl}}}}
\def \TurnN {\mathrel{\Yellow \turn\kern-2\point_{\ftsc{n}}}}
\def \TurnNf {\mathrel{\Yellow \turn\kern-2\point_{\ftrm{n}}}}
\def \TurnV {\mathrel{\Yellow \turn\kern-2\point_{\ftsc{v}}}}
\def \TurnVf {\mathrel{\Yellow \turn\kern-2\point_{\ftrm{v}}}}
\def \TurnX {\mathrel{\turn\kern-2\point_{\cal X}}}
\def \TurnP {\mathrel{\turn\kern-3\point_p}}
\def \TurnM {\mathrel{\turn\kern-4\point_{\mbox{\tinyfont $\cal M$}}}}
\def \TurnMc {\mathrel{\turn\kern-4\point_{\mbox{\tinyfont $\cal M^{\tinysc{c}}$}}}}
\def \TurnMe {\mathrel{\turn\kern-4\point_{\mbox{\tinyfont $\cal M^{\tinysc{e}}$}}}}
\def \TurnMIU {\mathrel{\turn\kern-4\point_{\mbox{\tinyfont $\cal M^{\cap\cup}$}}}}
\def \TurnMI {\mathrel{\turn\kern-4\point_{\mbox{\tinyfont $\cal M^{\cap}$}}}}

\def \MIU {{\Red \cal M^{\cap\cup}}}

\def \sclmmt {\mbox{
\tiny $\Overline{`l}`m\mt$}}

\def \Der#1#2{\@ifnextchar{\bgroup}{\DerThree{#1}{#2}}{\DerTwo{#1}{#2}}}
\def \DerTwo#1#2{{\typecolour #1 \Turn #2}}
\def \DerThree#1#2#3{{\typecolour #2 \Turn #3}}
\def \Stoup#1#2{\lstat{#1}{#2}}
\def \StoupR#1#2{\lstat{#1}{#2} \Mid}
\def \StoupR#1#2{\framebox{$#1\,{:}\,#2$}~}
\def \StoupR#1#2{\StoupFrame{\lstat{#1}{#2}}}

\def \Stoup#1#2{\lstat{#1}{#2}}

\def \StoupR#1#2{\lstat{#1}{#2}\Skip\Mid}
\def \DerLmuThree#1#2#3{#1 \TurnLmu \Stoup{#2}{#3} }
\def \DerLmuFour#1#2#3#4{#1 \TurnLmu \StoupR{#2}{#3} #4}
\def \DerLmu#1#2#3{\@ifnextchar{\bgroup}{\DerLmuFour{#1}{#2}{#3}}{\DerLmuThree{#1}{#2}{#3}}}
\def \derLmuFour #1 |- #2 : #3 | #4 {\type{#1} \TurnLmu \Stoup{#2}{#3} \Mid \type{#4} }
\def \derLmuNormal #1 |- #2 : #3 {\@ifnextchar|{\derLmuFour #1 |- #2 : #3 }{#1 \TurnLmu \Stoup{#2}{#3} }}
\def \derLmuCmd #1 |- #2 | #3 {\type{#1} \TurnLmu \lterm{#2} \Mid \type{#3} }
\def \derLmuBasis #1 |- #2 {\@ifnextchar:{\derLmuNormal #1 |- #2 }{\derLmuCmd #1 |- #2 }}
\def \derLmu {\@ifnextchar|{\derLmuBasis {\emptyset} }{\derLmuBasis }}

\def \derLmuIU #1 {\def \TurnLmu {\TurnLmuIU}\derLmu {#1} }

\def \derMCLFour #1 |- #2 | #3 {#1 \TurnMCL {#2} \mid {#3} }
\def \derMCLNormal #1 |- #2 {\@ifnextchar|{\derMCLFour #1 |- #2 }{#1 \TurnMCL {#2} }}
\def \derdeMCL #1 |- | #2 {#1 \TurnMCL \mid {#2} }
\def \derMCL #1 |- {\@ifnextchar|{\derdeMCL #1 |- }{\derMCLNormal {#1} |- }}
\def \DerLmmt#1#2{\@ifnextchar{\bgroup}{\DerLmmtThree{#1}{#2}}{\DerLmmtTwo{#1}{#2}}}
\def \DerLmmtTwo#1#2{{\typecolour #1 \TurnLmmt #2}}
\def \DerLmmtThree#1#2#3{#1~{:}~{\typecolour #2 \TurnLmmt #3}}

\def \derLmmtTwo #1 |- #2 {{\typecolour {#1} \TurnLmmt {#2}}}
\def \derLmmtCommand #1 : #2 |- #3 {{\ltermcolour {#1}}~{\black :}~{\typecolour {#2} \TurnLmmt {#3}}}
\def \derLmmtTerm #1 |- #2 : #3 | #4 {{\typecolour {#1} \TurnLmmt {\ltermcolour #2}\skipper{\black :}\skipper{#3} ~{\black \mid}~ {#4} }}
\def \derLmmtContext #1 | #2 : #3 |- #4 {{\typecolour {#1}~{\black \mid}~{\ltermcolour #2}\skipper{\black :}\skipper{#3} \TurnLmmt {#4} }}
\def \derLmmtActivated #1 |{\@ifnextchar-{\derLmmtTerm {#1} |}{\derLmmtContext {#1} | }}
\def \derLmmtNamedDer #1 :: {{#1} \dcol \derLmmt }
\def \derLmmtColon #1 :{\@ifnextchar:{\derLmmtNamedDer #1 :}{\derLmmtCommand {#1} : }}
\def \derLmmt #1 {\@ifnextchar:{\derLmmtColon {#1} }{\derLmmtActivated {#1} }}
\def \derLmmtV #1 {\def \TurnLmmt{\TurnLmmtV}\derLmmt {#1} }
\def \derLmmtN #1 {\def \TurnLmmt{\TurnLmmtN}\derLmmt {#1} }
\def \derLmmtSN #1 {\def \TurnLmmt{\TurnLmmtSN}\derLmmt {#1} }
\def \derLmmtmin #1 {\def \TurnLmmt{\TurnLmmtmin} \derLmmt {#1} }
\def \derLmmtM #1 {\def \TurnLmmt{\TurnM} \derLmmt {#1} }
\def \derLmmtMc #1 {\def \TurnLmmt{\TurnMc} \derLmmt {#1} }

\def \derL #1 |- #2 : #3 {{\typecolour #1 \TurnL {\ltermcolour #2}\,{\black :}\,#3}}
\def \derLx #1 |- #2 : #3 {{\typecolour #1 \TurnLx {\ltermcolour #2}\,{\black :}\,#3}}
\def \derI #1 |- #2 : #3 {{\typecolour #1 \TurnI {\ltermcolour #2}\,{\black :}\,#3}}
\def \der #1 |- #2 : #3 {{\typecolour #1 \TurnSimple {\ltermcolour #2}\,{\black :}\,#3}}

\def \derLK #1 |- #2 {{\typecolour #1}	\TurnLK	{\typecolour #2}}
\def \derLKarr #1 |- #2 {{\typecolour #1}	\TurnLKarr	{\typecolour #2}}
\def \deriv #1 |- #2 {{\typecolour #1}	\Turn	{\typecolour #2}}

\setbox161=\hbox{$\cdot$}
\setbox162=\hbox{{\neutral ~\raise-2.5\point\copy161\kern-\wd161\raise2.5\point\copy161\kern.5\point \raise0\point\copy161~}}
\def \DP {\copy162}
\def \DPa {{\neutral \raise-.7\point\hbox{~\large :~}}}

\def \PDer#1#2#3{{\termcolour #1} \DP {\typecolour #2} \Turn {\typecolour #3}}
\def \pDer#1:#2 |- #3 {{\termcolour #1} \DP {\typecolour {#2} \Turn {#3}}}
\def \InfPDerFour#1#2#3#4{ \InfBox{#1}{ \PDer{#2}{#3}{#4} } }
\def \InfPDer#1#2#3{\@ifnextchar{\bgroup}{\InfPDerFour{#1}{#2}{#3}}{\InfPDerFour{}{#1}{#2}{#3}}}
\def \PDerN#1#2#3{\typecolour {\termcolour #1 }\DP {\typecolour #2} \TurnN {\typecolour #3}}

\def \derXaux #1 : #2 |- #3 {{\termcolour #1 } \DP {\typecolour #2}	\TurnX	{\typecolour #3}}
\def \derX #1 : {\@ifnextchar|{\derXaux {#1} : {} }{\derXaux {#1} : }}

\def \derXR #1 : #2 |- #3 {{\termcolour #1 } \DP {\typecolour #2}	\TurnR	{\typecolour #3}}
\def \derXRN #1 : #2 |- #3 {{\termcolour #1 } \DP {\typecolour #2}	\TurnLN	{\typecolour #3}}
\def \derXRV #1 : #2 |- #3 {{\termcolour #1 } \DP {\typecolour #2}	\TurnLV	{\typecolour #3}}
\def \derXN #1 : #2 |- #3 {{\termcolour #1 } \DP {\typecolour #2}	\TurnN	{\typecolour #3}}
\def \derXNf #1 : #2 |- #3 {{\termcolour #1 } \DP {\typecolour #2}	\TurnNf	{\typecolour #3}}
\def \derXV #1 : #2 |- #3 {{\termcolour #1 } \DP {\typecolour #2}	\TurnV	{\typecolour #3}}
\def \derXVf #1 : #2 |- #3 {{\termcolour #1 } \DP {\typecolour #2}	\TurnVf	{\typecolour #3}}
\def \PDerV#1#2#3{\typecolour {\termcolour #1 }\DP {\typecolour #2} \TurnV {\typecolour #3}}
\def \InfPDerFourN#1#2#3#4{ \InfBox{#1}{ \PDerN{#2}{#3}{#4} } }
\def \InfPDerN#1#2#3{\@ifnextchar{\bgroup}{\InfPDerFourN{#1}{#2}{#3}}{\InfPDerFour{}{#1}{#2}{#3}}}
\def \InfPDerFourV#1#2#3#4{ \InfBox{#1}{ \PDerV{#2}{#3}{#4} } }
\def \InfPDerV#1#2#3{\@ifnextchar{\bgroup}{\InfPDerFourV{#1}{#2}{#3}}{\InfPDerFour{}{#1}{#2}{#3}}}

\def \derPure #1 : #2 |- #3 {{\termcolour #1 } \DP {\typecolour #2}	\Turn_{\kern-2pt\ftsc{p}}	{\typecolour #3}}
\def \derPureN #1 : #2 |- #3 {{\termcolour #1 } \DP {\typecolour #2}	\Turn_{\kern-2pt\ftsc{pn}}	{\typecolour #3}}
\def \derPureV #1 : #2 |- #3 {{\termcolour #1 } \DP {\typecolour #2}	\Turn_{\kern-2pt\ftsc{pv}}	{\typecolour #3}}

\def \derDP #1 : #2 |- #3 {{\termcolour #1 } \DP {\typecolour #2}	\Turn	{\typecolour #3}}
\def \LogicDer #1 |- #2 { {\typecolour #1}	\Turn	{\typecolour #2}}

\def \InfDerX #1 :: #2 : #3 |- #4 {\InfPDer{#1}{#2}{#3}{#4}}
\def \InfDer #1 :: #2 |- #3 : #4 {\InfLDer{#1}{#2}{#3}{#4}}

\def \XIn#1#2{{\widehat{#1}}\hspace*{.75pt}{#2}}
\def \XOut#1#2{{#1}\hspace*{1pt}{\widehat{#2}}}

\newfont{\oldmath}{cmsy10 at 11pt}
\newfont{\largeoldmath}{cmsy10 at 15pt}
\def \Dagger	{\mathrel{\mbox{\oldmath\symbol{121}}}}
\def \RedDagger	{\mathrel{\raise-1\point\hbox{\lightredcol \largeoldmath \symbol{121}}}}
\def \DaggerA{{\Yellow \Dagger\kern-3\point_{\mbox{\scriptsize \sc a}}}}
\def \DaggerLog{\kern-1\point{\Yellow \Dagger\kern-1.5\point_{\mbox{\scriptsize \sc l}\kern-1.5\point}}}

\def \DaggerLeft{\mathop{\kern-1\point{\specialcol \raise2\point\hbox{\rotatebox{-30}{$\Dagger$}}}\kern-1\point}}
\def \DaggerRight{\mathop{\kern-1\point{\specialcol \raise0\point\hbox{\rotatebox{30}{$\Dagger$}}}\kern-1\point}}

\def \DaggerV{\mathop{\,\Yellow \Dagger_{\mbox{\scriptsize \sc v}}}}
\def \DaggerN{\mathop{\,\Yellow \Dagger_{\mbox{\scriptsize \sc n}}}}
\def \DaggerL{\DaggerLeft}
\def \DaggerR{\DaggerRight}

\def \expT#1#2#3#4{\XOut{\XIn{#1}{#2}}{#3}\kern.2mm{`.}\kern.2mm#4}
\def \medT#1#2#3#4#5{\setbox111=\hbox{$#1$}\setbox112=\hbox{$#5$}%
\ifdim\wd111<10pt%
	\ifdim\wd112<10pt%
		\XOut{\box111}{#2}\,[#3]\,\XIn{#4}{\box112}%
	\else
		\XOut{\box111}{#2}~[#3]~\XIn{#4}{\box112}
	\fi
\else
	\XOut{\box111}{#2}~[#3]~\XIn{#4}{\box112}
\fi}
\def \medTdl#1#2#3#4#5{{\termcolour %
	\XOut{(#1)}{#2}~[#3] }\quad \\ \hfill {\termcolour \XIn{#4}{(#5)} }
}

\def \cutT#1#2#3#4{\XOut{#1}{#2} \Dagger \XIn{#3}{#4}}
\def \ColcutT#1#2#3#4{\XOut{#1}{#2} \RedDagger \XIn{#3}{#4}}
\def \cutlog#1#2#3#4{\XOut{#1}{\orange #2}{\orange \Dagger \XIn{#3}{#4}}}
\def \cutA#1#2#3#4{\XOut{#1}{#2}\DaggerA \XIn{#3}{#4}}
\def \invLT#1#2#3{{#1}\kern.2mm{`.}\kern.2mm\XOut{#2}{#3}}
\def \invRT#1#2#3{\XIn{#1}{#2}\kern.2mm{`.}\kern.2mm{#3}}
\def \negL #1 . #2 #3 {\NegL{#1}{#2}{#3}}
\def \negR #1 #2 . #3 {\NegR{#1}{#2}{#3}}
\def \negLT#1#2#3{{#1}\kern.2mm{`.}\kern.2mm\XOut{#2}{#3}}
\def \negRT#1#2#3{\XIn{#1}{#2}\kern.2mm{`.}\kern.2mm{#3}}
\def \pairT#1#2#3#4#5{\Group<{#1}\widehat{#2},{#3}\widehat{#4}>\kern.2mm{`.}\kern.2mm{#5}}
\def \orLT#1#2#3#4#5{{#1}\kern.2mm{`.}\kern.2mm (\widehat{#2}{#3}\kern.2mm{+}\kern.2mm\widehat{#4}{#5}) }
\def \inLT#1#2#3#4{{#1}(\widehat{#2}\kern.2mm{+}\kern.2mm{#3})\kern.2mm{`.}\kern.2mm{#4}}
\def \inRT#1#2#3#4{{#1}\,({#2}\kern.2mm{+}\kern.2mm\widehat{#3})\kern.2mm{`.}\kern.2mm{#4}}
\def \orRPT#1#2#3#4{{#1}\,(\widehat{#2}\kern.3mm{\mid}\kern.3mm\widehat{#3})\kern.2mm{`.}\kern.2mm{#4}}

\def \newconRT#1#2#3#4{\Group<{#1},{#2}>{#3}\kern.2mm{`.}\kern.2mm{#4}}
\def \newconLT#1#2#3#4#5{#1\kern.2mm{`.}\kern.2mm\Group<{#2}\widehat{#3},{#4}\widehat{#5}>}

\def \Weak{\mysl{W}}
\def \Thin{\mysl{T}}

\omtrue
\def \InvLterm#1#2#3{{\termcolour \if@om \omfalse\invLT{#1}{#2}{#3}\omtrue
	\else	\skp(\invLT{#1}{#2}{#3})\fi}}
\def \InvRterm#1#2#3{{\termcolour \if@om \omfalse\invRT{#1}{#2}{#3}\omtrue
	\else	\skp(\invRT{#1}{#2}{#3})\fi}}
\def \NegL#1#2#3{{\termcolour \if@om \omfalse\negLT{#1}{#2}{#3}\omtrue
	\else	\skp(\negLT{#1}{#2}{#3})\fi}}
\def \NegR#1#2#3{{\termcolour \if@om \omfalse\negRT{#1}{#2}{#3}\omtrue
	\else	\skp(\negRT{#1}{#2}{#3})\fi}}
\def \NewConR#1#2#3#4{{\termcolour \if@om \omfalse\newconRT{#1}{#2}{#3}{#4}\omtrue	\else	(\newconRT{#1}{#2}{#3}{#4})\fi}}
\def \NewConL#1#2#3#4#5{{\termcolour \if@om \omfalse\newconLT{#1}{#2}{#3}{#4}{#5}\omtrue	\else	(\newconLT{#1}{#2}{#3}{#4}{#5})\fi}}

\def \nonconnectable{\mbox{\tiny $\Box$}}
\def \nonconnectable{\diamond}
\def \newconR < {\@ifnextchar{o}{\newconRR < }{\newconRL < }}
\def \newconRL < #1 , #2 > #3 . #4 {\NewConR{\widehat{#1}}{\nonconnectable}{#3}{#4}}
\def \newconRR < #1 , #2 > #3 . #4 {\NewConR{\nonconnectable}{\widehat{#2}}{#3}{#4}}

\def \newconL #1 . < #2 #3 , #4 #5 > {\NewConL{#1}{#2}{#3}{#4}{#5}}

\newlength\harpoonlength

\def \longharpoon#1{\psset{unit=.7\point,linewidth=0.4\point,linecolor=ltermcol}%
\harpoonlength#1\addtolength{\harpoonlength}{-1\point}%
\psline{cc-cc}(-1,1.25)(-\harpoonlength,1.25)%
\pscurve{cc-cc}(-0.7,1.25)(-1.5,1.44)(-2,1.8)(-2.5,2.3)(-3,3.1)
\vbox to 2\point{}%
}

\def \Vec#1{\setbox155=\hbox{$#1$}%
\leavevmode\copy155\raise\ht155\hbox{\longharpoon{\wd155}}}

\def \longarrow#1{\psset{unit=1\point,linewidth=0.35\point}%
\psline{cc-cc}(0,0)(#1,0)%
\hspace*{#1}%
\pscurve{cc-cc}(0,0)(-1.1,.4)(-2.5,2.2)%
\pscurve{cc-cc}(0,0)(-1.1,-.4)(-2.5,-2.2)}

\def \Rel#1#2{\setbox155=\hbox{\scriptsize $#1$}\setbox156=\hbox{$#2$}%
\mathrel{\raise5.5pt\copy155\kern-1\wd155%
\kern-.5\wd155\kern-.5\wd156\raise-2pt\copy156\kern-.5\wd156\kern.5\wd155}}

\def \stackrel#1#2{\setbox155=\hbox{\scriptsize $#1$}\setbox156=\hbox{$#2$}%
\mathrel{\copy156\kern-.5\wd156\kern-.5\wd155\raise4.5\point\hbox{\copy155}\kern-.5\wd155\kern.5\wd156}}

\def \Overline#1{%
 \ifmycolour\psset{unit=1\point,linewidth=0.4\point,linecolor=termcol}%
 \else\psset{unit=1\point,linewidth=0.35\point,linecolor=black}
 \fi
 \setbox155=\hbox{$#1$}\leavevmode
 \raise\ht155\hbox{\psline{c-c}(.05\wd155,1.25)(.95\wd155,1.25)}\box155%
}
\def \Underline#1{%
 \ifmycolour\psset{unit=1\point,linewidth=0.4\point,linecolor=yellow}%
 \else\psset{unit=1\point,linewidth=0.35\point,linecolor=black}
 \fi
 \setbox156=\hbox{${#1}$}\leavevmode
 \raise-\dp156\hbox{\psline{c-c}(.05\wd156,-0.75)(.95\wd156,-0.75)}\box156 }

\def \ul#1{\Underline{\mbox{\scriptsize $#1$}}}

\def \CH#1{\kern1\point\lefttop{\mbox{\Purple $#1$}}\righttop\kern1\point}
\def \LtoLMMT#1{\kern1\point\lefttop{#1}\righttop\kern1\point}
\def \CH#1{\raise1.5\point\hbox{\footnotesize ${\mid}\kern-1.5\point[$}%
#1
\raise1.5\point\hbox{\footnotesize $]\kern-1.5\point{\mid}_{`l}$}}

\def \CutGraph{\@ifnextchar[{\CutGraphN}{\CutGraphN[\GCut]}}

\def \CutGraphN[#1]#2#3#4#5{%
\begin{psmatrix}[rowsep=1ex,colsep=.05ex]
&&&[name=cutroot]&\\[3ex]
&&&[name=cut]#1&\\
&[name=P]#2&&&&[name=Q]#5& \\
&&[name=a]#3&&[name=x]#4\\&
\end{psmatrix}
 \ncline{cutroot}{cut}
\ncline{cut}{P}
\ncline{cut}{a}
\ncline{cut}{x}
\ncline{cut}{Q}
}

\def \V#1{\setbox71=\hbox{#1}\setbox72=\hbox{\v{}}%
\copy71\kern-.5\wd71\kern-.3\wd72\raise\ht71\hbox{\raise.2em\hbox{\lower\ht72\hbox{\copy72}}}\kern-.7\wd72\kern.5\wd71}

\setbox138=\hbox{$=$}
\setbox139=\hbox{{\scriptsize $\Delta$}}%
\setbox149=\hbox{\raise-1\point\copy138\kern-.5\wd138\kern-.5\wd139\raise4\point\hbox{\copy139}\kern-.5\wd139\kern.5\wd138}
\def \ByDef {\mathrel{\copy149}}

\def \Cut{\@ifnextchar{\bgroup}{\CutTwo}{{\mysl{cut}}}}

\def \Exp	{\@ifnextchar{\bgroup}{\ExpT}{{\mysl{exp}}}}
\def \Imp	{\@ifnextchar{\bgroup}{\MedT}{{\mysl{imp}}}}
\def \Med	{\@ifnextchar{\bgroup}{\MedT}{{\mysl{med}}}}
\def \Cap	{\@ifnextchar{\bgroup}{\CapT}{{\mysl{cap}}}}

\def \CutL#1#2{\@ifnextchar{\bgroup}{\CutLeft{#1}{#2}}{\widehat{#1} \DaggerL \widehat{#2}}}
\def \CutR#1#2{\@ifnextchar{\bgroup}{\CutRight{#1}{#2}}{\widehat{#1} \DaggerR \widehat{#2}}}
\def \CutTwo#1#2{\@ifnextchar{\bgroup}{\CutT{#1}{#2}}{\widehat{#1} \Dagger \widehat{#2}}}
\def \ColCut#1#2{\@ifnextchar{\bgroup}{\ColCutT{#1}{#2}}{\widehat{#1} \Dagger \widehat{#2}}}
\def \Cutdl#1#2#3#4{\XOut{(#1)}{#2} \Dagger\\ \hfill \XIn{#3}{(#4)}}

\def \TrTwo#1#2{\Sem{\ltermcolour #1}{\termcolour #2}^{\sclmmt}}
\def \Tran#1{\@ifnextchar{\bgroup}{\TrTwo{#1}}{\TrTwo{#1}{}}}

\newdimen \GeneWidth \newdimen \GeneW \newdimen \HGeneW
\newdimen \GeneHeight \newdimen \GeneH \newdimen \HGeneH
\newcount \PictW \newcount \HPictW
\newcount \PictH \newcount \HPictH
\newcount \BoxH \newcount \HBoxH
\newcount \BoxW \newcount \HBoxW
\newcount \HOffset \newcount \VOffset

\def \StoupFrameR#1{\setbox159=\hbox{$#1$}%
\PictW\wd159\divide\PictW\unitlength\advance\PictW3%
\PictH\ht159\multiply\PictH18\divide\PictH10\divide\PictH\unitlength
\,\raise-4\unitlength\hbox{\PSFrame(\PictW,\PictH)}%
\kern1\point\box159\,,}

\def \StoupFrameL#1{\setbox159=\hbox{$#1$}%
\PictW\wd159\divide\PictW\unitlength\advance\PictW3%
\PictH\ht159\multiply\PictH18\divide\PictH10\divide\PictH\unitlength
,\,\raise-4\unitlength\hbox{\PSFrame(\PictW,\PictH)}%
\kern1\point\box159\,}

\def \Gene#1{\unitlength.1\point\setbox140=\hbox{$#1$}%
\GeneHeight\ht140\advance\GeneHeight55\unitlength%
\GeneWidth\wd140\advance\GeneWidth50\unitlength%
\PictW\GeneWidth\divide\PictW\unitlength%
\PictH\GeneHeight\divide\PictH\unitlength%
\ifdim\wd140>10\point
\begin{picture}(\PictW,\PictH)%
 \put(0,0)	{\makebox(\PictW,\PictH){\box140}}%
 \put(0,0)	{\PSFrame(\PictW,\PictH)}%
\end{picture}
\else \box140
\fi}

\def \InvGene#1{\unitlength.1\point%
\setbox140=\hbox{$#1$}%
\GeneHeight\ht140\advance\GeneHeight60\unitlength
\GeneWidth\wd140\advance\GeneWidth50\unitlength
\PictW\GeneWidth\divide\PictW\unitlength
\PictH\GeneHeight\divide\PictH\unitlength
\begin{picture}(\PictW,\PictH)
 \put(0,0)	{\makebox(\PictW,\PictH){\box140}}
 \put(0,0)	{\PSFrame(0,0)}
\end{picture}}

\newcommand{\DiagCap}[2]{\unitlength.1\point%
\begin{picture}(360,120)(-180,-55)
	\put(-180,0)	{\Yellow \vector(1,0){140}}
	\put(-110,50)	{\makebox(0,0){\small $\termcolour #1$}}
	\put(-40,-30)	{\PSBox(80,80)}
	\put(110,50)	{\makebox(0,0){\small $\termcolour #2$}}
	\put(40,0)	{\Yellow	\vector(1,0){140}}
 \end{picture}}
\def \Diagcaps<#1,#2>{\DiagCap{#1}{#2}}

\newcommand{\DiagExp}[4]{\unitlength.1\point%
\setbox143=\hbox{{\termcolour \small $#2$}}%
\GeneWidth\wd143
	\advance\GeneWidth320\unitlength
	\BoxW\GeneWidth \divide\BoxW\unitlength \HBoxW\BoxW \divide\HBoxW2%
	\advance\GeneWidth200\unitlength \PictW\GeneWidth \divide\PictW\unitlength \HPictW\PictW \divide\HPictW2%
\GeneHeight\ht143 \advance\GeneHeight 80\unitlength
	\BoxH\GeneHeight \divide\BoxH\unitlength \HBoxH\BoxH \divide\HBoxH2%
	\advance\GeneHeight70\unitlength
\PictH\GeneHeight \divide\PictH\unitlength \HPictH\PictH \divide\HPictH2%
\HOffset\wd143 \divide\HOffset\unitlength \divide\HOffset2%
\VOffset\HPictW \divide\VOffset\unitlength \divide\VOffset2 \advance\VOffset20%
 \def \ContentExp {\begin{picture}(\PictW,\PictH)(-\HPictW,-\HPictH)
 \put(-50,0){%
 \VOffset\ht143\divide\VOffset\unitlength\divide\VOffset2\advance\VOffset50%
 \makebox(0,0){\Gene{\begin{picture}(\BoxW,\BoxH)(-\HBoxW,20)
		\put(-\HOffset,\VOffset)	{
		\put(-160,-25) {\Yellow \vector(1,0){150}}
		\put(-85,22)	{\makebox(0,0){\termcolour \small $\widehat{#1}$}}
	}%
	\put(0,\VOffset)	{\makebox(0,0){\box143}}
		\put(\HOffset,\VOffset) {
		\put(10,-25)	{\Yellow \vector(1,0){150}}
		\put(75,20)	{\makebox(0,0){\termcolour \small $\widehat{#3}$}}
	}
 \end{picture}}}}
 \HOffset\HBoxW\advance\HOffset-40
 \put(\HOffset,-15)	{\Yellow \vector(1,0){130}}
 \advance\HOffset65
 \put(\HOffset,35) {\makebox(0,0){\termcolour \small $#4$}}
 \end{picture}}
\if@om \omfalse 
\ContentExp 
\omtrue \else \Gene{\ContentExp} \fi
}

\newdimen \MedW \newdimen \MedH
\newcommand{\DiagMed}[5]{\unitlength.1\point%
\setbox144=\hbox{{\termcolour \small $#1$}}%
\setbox145=\hbox{{\termcolour \small $#5$}}%
\setbox146=\hbox{\termcolour \small $#3$}%
\MedW \wd144 \advance \MedW \wd145 
\advance \MedW 480\unitlength
\ifdim \ht144 > \ht145 \MedH \ht144 \else \MedH \ht145 \fi \advance \MedH 70\unitlength
\BoxW\MedW\divide\BoxW\unitlength\HBoxW\BoxW\divide\HBoxW2%
\BoxH\MedH\divide\BoxH\unitlength\HBoxH\BoxH\divide\HBoxH2%
\GeneWidth \MedW \advance \GeneWidth 240\unitlength
\GeneHeight \MedH \advance \GeneHeight 70\unitlength
\PictW \GeneWidth \divide \PictW \unitlength \HPictW \PictW \divide \HPictW 2%
\PictH \GeneHeight \divide \PictH \unitlength \HPictH \PictH \divide \HPictH 2%
\def \ContentMed{\begin{picture}(\PictW,\PictH)(-\HPictW,-\HPictH)
 \put(75,0){%
	\makebox(0,0){%
	\VOffset \HPictH \advance \VOffset -3%
	\HOffset \MedW 
		\advance \HOffset \wd144
		\advance \HOffset -\wd145 
	\divide \HOffset \unitlength \divide \HOffset 2
 \VOffset \HBoxH \advance \VOffset 0
 \Gene{\begin{picture}(\BoxW,\BoxH)(-\HOffset,-\VOffset)
	\put(-220,0) {\makebox(0,0)[r]{\box144}}
	\put(-180,-20) {\Yellow \vector(1,0){120}}
	\put(-130,30)	{\makebox(0,0){\termcolour \small $\widehat{#2}$}}
	\put(0,0)	{\makebox(0,0){${\Yellow [~]}$}}
	\put(120,30)	{\makebox(0,0){\termcolour \small $\widehat{#4}$}}
	\put(60,-20)	{\Yellow \vector(1,0){120}}
	\put(220,0)	{\makebox(0,0)[l]{\box145}}
	\end{picture}}}}
 \HOffset\HBoxW\advance\HOffset120
 \put(-\HOffset,-15) {\Yellow \vector(1,0){140}}
 \advance\HOffset-70
 \put(-\HOffset,35) {\makebox(0,0){\termcolour $#3$}}
 \end{picture}}
\ContentMed 
}
\def \Diagmed #1 #2 [#3] #4 #5 {\DiagMed{#1}{#2}{#3}{#4}{#5}}
 
\newcommand{\DiagDMed}[5]{\unitlength.1\point%
\setbox144=\hbox{{\termcolour $#1$}}\setbox145=\hbox{{\termcolour $#5$}}\setbox146=\hbox{\termcolour $#3$}%
\ifdim \wd144 > \wd145 \MedW \wd144 \else \MedW \wd145 \fi 
\advance \MedW 300\unitlength
\MedH \ht144 \advance \MedH \ht145 \advance \MedH 100\unitlength
\BoxW\MedW\divide\BoxW\unitlength\HBoxW\BoxW\divide\HBoxW2%
\BoxH\MedH\divide\BoxH\unitlength\HBoxH\BoxH\divide\HBoxH2%
\GeneWidth \MedW \advance \GeneWidth 240\unitlength
\GeneHeight \MedH \advance \GeneHeight 70\unitlength
\PictW \GeneWidth \divide \PictW \unitlength \HPictW \PictW \divide \HPictW 2%
\PictH \GeneHeight \divide \PictH \unitlength \HPictH \PictH \divide \HPictH 2%
\def \ContentDMed{\begin{picture}(\PictW,\PictH)(-\HPictW,-\HPictH)%
 \put(80,0){%
	\makebox(0,0){%
	\VOffset \HPictH \divide \VOffset 2%
	\ifdim \wd144 > \wd145 \HOffset \wd144
	\else \HOffset \wd145 \fi
	\divide \HOffset \unitlength \advance \HOffset 50%
 \Gene{\begin{picture}(\BoxW,\BoxH)(-\HOffset,-\HBoxH)
	\put(0,\VOffset){%
	\put(0,-20) {\makebox(0,0)[r]{\box144}}
	\put(40,-30) {\Yellow \vector(1,0){120}}
	\put(95,18)	{\makebox(0,0){\termcolour $\widehat{#2}$}}
	}
	\put(-\HOffset,-\VOffset){%
	\put(80,0)	{\makebox(0,0){${\Yellow [~]}$}}
	\put(175,18)	{\makebox(0,0){\termcolour $\widehat{#4}$}}
	\put(125,-30)	{\Yellow \vector(1,0){120}}
	\put(250,0)	{\makebox(0,0)[l]{\box145}}
	}
	\end{picture}}}}%
 \HOffset\HBoxW\advance\HOffset120%
 \put(-\HOffset,-20) {\Yellow \vector(1,0){140}}%
 \advance\HOffset-70%
 \put(-\HOffset,30) {\makebox(0,0){\termcolour $#3$}}%
 \end{picture}}%
 \ContentDMed 
}

\newcommand{\DiagCutA}[5]{\unitlength.1\point%
\setbox141=\hbox{\small $#1$}%
\setbox142=\hbox{\small $#4$}%
\GeneWidth \wd141 \advance \GeneWidth \wd142 \advance \GeneWidth 400\unitlength
\ifdim \ht141 > \ht142 \GeneHeight \ht141 \else \GeneHeight \ht142 \fi \advance \GeneHeight 30\unitlength
\PictW \GeneWidth \divide \PictW \unitlength \HPictW \PictW \divide \HPictW 2%
\PictH \GeneHeight \divide \PictH \unitlength \HPictH \PictH \divide \HPictH 2 \advance \PictH 20%
\HOffset \GeneWidth
\ifdim \wd142 > \wd141 \advance \HOffset -\wd142 \advance \HOffset \wd141
\else \advance \HOffset \wd141 \advance \HOffset -\wd142 \fi
\divide \HOffset \unitlength \divide \HOffset 2
\VOffset \HPictH \advance \VOffset 0
\def \ContentCut{ \begin{picture}(\PictW,\PictH)(-\HOffset,-\VOffset)
	\put(-178,0)	{\makebox(0,0)[r]{\termcolour \box141}}
	\put(-140,-20)	{\Yellow \vector(1,0){160}}
	\put(-70,28)	{\makebox(0,0){\termcolour $\widehat{#2}$}}
	\put(0,-100)	{\makebox(0,0){\mbox{\Yellow \footnotesize \sc #5}}}
	\put(70,28)	{\makebox(0,0){\termcolour $\widehat{#3}$}}
	\put(0,-20)	{\Yellow \line(1,0){140}}
	\put(180,0)	{\makebox(0,0)[l]{\termcolour \box142}}
 \end{picture}}
\Gene{ \ContentCut }
}

\newcount\tempcount
\newcommand{\XNat}[1]{\tempcount#1\raise-4\point\hbox{\DiagCap{x}{`a}}\loop{\ifnum\tempcount>0%
\,[f]\,\raise-4\point\hbox{\DiagCap{x}{`a}}%
\advance\tempcount by-1 }\repeat}

\def \typeofargs#1#2{{\It {typeof}}\ #1\ #2}
\def \typeof{\@ifnextchar\bgroup{\typeofargs}{\It {typeof}}}

\def \arrow {{\rightarrow}}

\def \ftsc#1{\textsc{\scriptsize #1}}
\def \tinysc#1{\textsc{\tiny #1}}
\def \ftrm#1{\textrm{\scriptsize #1}}

\def \unifyContextsargs#1#2{{\It {unifyContexts}}\ #1\ #2}
\def \unifyContexts{\@ifnextchar\bgroup{\unifyContextsargs}{\It {unifyContexts}}}

\def \GCut {\hbox{\sf \small Cut}}

\def \LC	{\mbox{$\Red `l$-calculus}}

\def \DBs {\textsf{\small DB}}
\def \DBf#1#2#3{\DBs(#1,#2,#3)}
\def \DB {\@ifnextchar{\bgroup}{\DBf}{\DBs}}

\omtrue
\def \expT#1#2#3#4{{\typecolour \XOut{\XIn{#1}{#2}}{#3}\cdot #4}}
\def \medT#1#2#3#4#5{{\typecolour \XOut{#1}{#2}\,[#3]\,\XIn{#4}{#5}}}

\def \expT#1#2#3#4{\XOut{\XIn{#1}{#2}}{#3}\kern.2mm{`.}\kern.2mm#4}

\def \medT#1#2#3#4#5{\setbox111=\hbox{$#1$}\setbox112=\hbox{$#5$}%
\ifdim\wd111<10pt%
	\ifdim\wd112<10pt%
	\XOut{\box111}{#2}\,[#3]\,\XIn{#4}{\box112}%
	\else
	\XOut{\box111}{#2}~[#3]~\XIn{#4}{\box112}
	\fi
\else
	\XOut{\box111}{#2}~[#3]~\XIn{#4}{\box112}
\fi}

\def \cutV#1#2#3#4{\XOut{#1}{#2} \DaggerV \XIn{#3}{#4}}
\def \cutN#1#2#3#4{\XOut{#1}{#2} \DaggerN \XIn{#3}{#4}}
\def \cutleft#1#2#3#4{\XOut{#1}{#2} \DaggerL \XIn{#3}{#4}}
\def \cutright#1#2#3#4{\XOut{#1}{#2} \DaggerR \XIn{#3}{#4}}
\def \cutA#1#2#3#4{\XOut{#1}{#2}\DaggerA \XIn{#3}{#4}}
\def \negLT#1#2#3{{#1}\kern.2mm{`.}\kern.2mm\XOut{#2}{#3}}
\def \negRT#1#2#3{\XIn{#1}{#2}\kern.2mm{`.}\kern.2mm{#3}}

\def \orLT#1#2#3#4#5{{#1}\kern.2mm{`.}\kern.2mm\Group<\widehat{#2}{#3}\kern.3mm{\mid}\kern.3mm\widehat{#4}{#5}>}
\def \orRLT#1#2#3#4{{#1}\Group<\widehat{#2}\kern.2mm{\mid}\kern.2mm{#3}>\kern.2mm{`.}\kern.2mm{#4}}
\def \orRRT#1#2#3#4{{#1}\Group<{#2}\kern.2mm{\mid}\kern.2mm\widehat{#3}>\kern.2mm{`.}\kern.2mm{#4}}
\def \orRPT#1#2#3#4{{#1}\Group<\widehat{#2}\kern.2mm{\mid}\kern.2mm\widehat{#3}>\kern.2mm{`.}\kern.2mm{#4}}

\def \Caps(#1,#2){{\termcolour \langle#1{`.}#2\rangle}}
\def \CapT#1#2{\Caps(#1,#2)}
\def \ExpT#1#2#3#4{{\termcolour \if@om \omfalse\expT{#1}{#2}{#3}{#4}\omtrue
	\else (\expT{#1}{#2}{#3}{#4}) \fi}}
\def \Exps#1#2#3.#4{{\termcolour \if@om \omfalse\expT{#1}{#2}{#3}{#4}\omtrue
	\else (\expT{#1}{#2}{#3}{#4})\fi}}
\def \MedT#1#2#3#4#5{{\termcolour \if@om \omfalse\medT{#1}{#2}{#3}{#4}{#5}\omtrue
	\else (\medT{#1}{#2}{#3}{#4}{#5})\fi}}
\def \Meddl#1#2#3#4#5{\if@om \omfalse\medTdl{#1}{#2}{#3}{#4}{#5}\omtrue
	\else (\medTdl{#1}{#2}{#3}{#4}{#5})\fi}
\def \Meds#1#2[#3]#4#5{{\termcolour \if@om \omfalse\medT{#1}{#2}{#3}{#4}{#5}\omtrue \else (\medT{#1}{#2}{#3}{#4}{#5})\fi}}
\def \CutT#1#2#3#4{{\termcolour \if@om \omfalse\cutT{#1}{#2}{#3}{#4}\omtrue
	\else	(\cutT{#1}{#2}{#3}{#4})\fi}}
\def \ColCutT#1#2#3#4{{\termcolour \if@om \omfalse\ColcutT{#1}{#2}{#3}{#4}\omtrue
	\else	(\ColcutT{#1}{#2}{#3}{#4})\fi}}
\def \CutLog#1#2#3#4{{\termcolour \if@om \omfalse\cutlog{#1}{#2}{#3}{#4}\omtrue
	\else	(\cutlog{#1}{#2}{#3}{#4})\fi}}
\def \CutLeft#1#2#3#4{{\termcolour \if@om \omfalse\cutleft{#1}{#2}{#3}{#4}\omtrue \else	(\cutleft{#1}{#2}{#3}{#4})\fi}}
\def \CutRight#1#2#3#4{{\termcolour \if@om \omfalse\cutright{#1}{#2}{#3}{#4}\omtrue \else	(\cutright{#1}{#2}{#3}{#4})\fi}}
\def \CutV#1#2#3#4{{\termcolour \if@om \omfalse\cutV{#1}{#2}{#3}{#4}\omtrue
	\else	(\cutV{#1}{#2}{#3}{#4} )\fi}}
\def \CutN#1#2#3#4{{\termcolour \if@om \omfalse\cutN{#1}{#2}{#3}{#4}\omtrue
	\else	(\cutN{#1}{#2}{#3}{#4} )\fi}}
\def \CutA#1#2#3#4{{\termcolour \if@om \omfalse\cutA{#1}{#2}{#3}{#4}\omtrue
	\else	(\cutA{#1}{#2}{#3}{#4})\fi}}

\def \NegL#1#2#3{{\termcolour \if@om \omfalse\negLT{#1}{#2}{#3}\omtrue
	\else	(\negLT{#1}{#2}{#3})\fi}}
\def \NegR#1#2#3{{\termcolour \if@om \omfalse\negRT{#1}{#2}{#3}\omtrue
	\else	(\negRT{#1}{#2}{#3})\fi}}
\def \AndL#1#2#3#4{{\termcolour \if@om \omfalse \andLPT{#1}{#2}{#3}{#4} \omtrue \else \skp (\andLPT{#1}{#2}{#3}{#4}) \skp \fi}}
\def \AndLL#1#2#3#4{{\termcolour \if@om \omfalse \andLLT{#1}{#2}{#3}{#4} \omtrue \else \skp (\andLLT{#1}{#2}{#3}{#4}) \skp\fi}}
\def \AndLR#1#2#3#4{{\termcolour \if@om \omfalse\andLRT{#1}{#2}{#3}{#4}\omtrue	\else	(\andLRT{#1}{#2}{#3}{#4})\fi}}
\def \AndLP#1#2#3#4{{\termcolour \if@om \omfalse\andLPT{#1}{#2}{#3}{#4}\omtrue	\else	(\andLPT{#1}{#2}{#3}{#4})\fi}}
\def \AndR#1#2#3#4#5{{\termcolour \if@om \omfalse\andRT{#1}{#2}{#3}{#4}{#5}\omtrue	\else	(\andRT{#1}{#2}{#3}{#4}{#5})\fi}}
\def \OrL#1#2#3#4#5{{\termcolour \if@om \omfalse\orLT{#1}{#2}{#3}{#4}{#5}\omtrue	\else	(\orLT{#1}{#2}{#3}{#4}{#5})\fi}}
\def \OrRL#1#2#3#4{{\termcolour \if@om \omfalse\orRLT{#1}{#2}{#3}{#4}\omtrue	\else	(\orRLT{#1}{#2}{#3}{#4})\fi}}
\def \OrRR#1#2#3#4{{\termcolour \if@om \omfalse\orRRT{#1}{#2}{#3}{#4}\omtrue	\else	(\orRRT{#1}{#2}{#3}{#4})\fi}}
\def \OrRP#1#2#3#4{{\termcolour \if@om \omfalse\orRPT{#1}{#2}{#3}{#4}\omtrue	\else	(\orRPT{#1}{#2}{#3}{#4})\fi}}

\def \NewConR#1#2#3#4{{\termcolour \if@om \omfalse\newconRT{#1}{#2}{#3}{#4}\omtrue	\else	(\newconRT{#1}{#2}{#3}{#4})\fi}}
\def \NewConL#1#2#3#4#5{{\termcolour \if@om \omfalse\newconLT{#1}{#2}{#3}{#4}{#5}\omtrue	\else	(\newconLT{#1}{#2}{#3}{#4}{#5})\fi}}

\def \caps<#1,#2>{\Cap{#1}{#2}}
\def \med #1 #2 [#3] #4 #5 {\Med{#1}{#2}{#3}{#4}{#5}}
\def \imp #1 #2 [#3] #4 #5 {\Med{#1}{#2}{#3}{#4}{#5}}
\def \impdl #1 #2 [#3] #4 #5 {\Meddl{#1}{#2}{#3}{#4}{#5}}
\def \exp #1 #2 #3 . #4 {\Exp{#1}{#2}{#3}{#4}}
\def \expl #1 #2 #3 #4 . #5 {{#1}:\Exp{#2}{#3}{#4}{#5}}
\def \cut #1 #2 + #3 #4 {\Cut{#1}{#2}{#3}{#4}}
\def \cutdl #1 #2 + #3 #4 {\Cutdl{#1}{#2}{#3}{#4}}
\def \cutLog #1 #2 + #3 #4 {\CutLog{#1}{#2}{#3}{#4}}
\def \cutL #1 #2 + #3 #4 {\CutL{#1}{#2}{#3}{#4}}
\def \cutR #1 #2 + #3 #4 {\CutR{#1}{#2}{#3}{#4}}
\def \colcut #1 #2 + #3 #4 {\ColCut{#1}{#2}{#3}{#4}}
\def \colcutR #1 #2 + #3 #4 {\ColCut{#1}{#2}{#3}{#4}}
\def \colcutL #1 #2 + #3 #4 {\ColCut{#1}{#2}{#3}{#4}}
\def \negL #1 . #2 #3 {\NegL{#1}{#2}{#3}}
\def \negR #1 #2 . #3 {\NegR{#1}{#2}{#3}}

\def \andLL #1 . < #2 , #3 > #4 {\AndLL{#1}{#2}{\nonconnectable}{#4}}
\def \andLP #1 . < #2 , #3 > #4 {\AndLP{#1}{#2}{#3}{#4}}
\def \andLR #1 . < #2 , #3 > #4 {\AndLR{#1}{\nonconnectable}{#3}{#4}}
\def \andLLaux #1 . < #2 , {\@ifnextchar{o}{\andLL #1 . < #2 , }{\andLP #1 . < #2 , }}
\def \andL #1 . < {\@ifnextchar{o}{\andLR {#1} . < }{\andLLaux {#1} . < }}
\def \andLL#1#2#3#4{{#1}\kern.2mm{`.}\kern.2mm\Group<\widehat{#2},{#3}>{#4}}
\def \andLR#1#2#3#4{{#1}\kern.2mm{`.}\kern.2mm\Group<{#2},\widehat{#3}>{#4}}
\def \andLLT#1#2#3#4{{#1}\kern.2mm{`.}\kern.2mm\Group<\widehat{#2},{#3}>{#4}}
\def \andLRT#1#2#3#4{{#1}\kern.2mm{`.}\kern.2mm\Group<{#2},\widehat{#3}>{#4}}
\def \andLPT#1#2#3#4{{#1}\kern.2mm{`.}\kern.2mm\Group<\widehat{#2},\widehat{#3}>{#4}}
\def \andRT#1#2#3#4#5{\Group<{#1}\widehat{#2},{#3}\widehat{#4}>\kern.2mm{`.}\kern.2mm{#5}}
\def \andRrule{\Conj\textit{R}}
\def \andRlog < #1 #2 , #3 #4 > . #5 {\AndR{#1}{#2}{#3}{#4}{#5}}
\def \andR{\@ifnextchar<{\andRlog}{\andRrule}}
\def \invL #1 . #2 #3 {\InvLterm{#1}{#2}{#3}}
\def \invR #1 #2 . #3 {\InvRterm{#1}{#2}{#3}}

\def \orR #1 < {\@ifnextchar{o}{\orRR {#1} < }{\orRLaux {#1} < }}
\def \orRLaux #1 < #2 | {\@ifnextchar{o}{\orRL #1 < #2 | }{\orRP #1 < #2 | }}
\def \orRL #1 < #2 | #3 > . #4 {\OrRL{#1}{#2}{\nonconnectable}{#4}}
\def \orRR #1 < #2 | #3 > . #4 {\OrRR{#1}{\nonconnectable}{#3}{#4}}
\def \orRP #1 < #2 | #3 > . #4 {\OrRP{#1}{#2}{#3}{#4}}
\def \projLT#1#2#3#4{{#1}\kern.2mm{`.}\kern.2mm\Group<\widehat{#2},{#3}>{#4}}
\def \projLR #1 . < #2 , #3 > #4 {\AndL{#1}{#2}{#3}{#4} }
\def \projL #1 . < #2 , #3 > #4 {\AndLL{#1}{#2}{\nonconnectable}{#4}}
\def \projLaux #1 . < #2 , {\@ifnextchar{o}{\projL {#1} . < {#2} , }{\projLR {#1} . < {#2} , }}
\def \projR #1 . < #2 , #3 > #4 {\AndLR{#1}{\nonconnectable}{#3}{#4}}
\def \proj #1 . < {\@ifnextchar{o}{\projR {#1} . < }{\projLaux {#1} . < }}
\def \pair < #1 #2 , #3 #4 > . #5 {\AndR{#1}{#2}{#3}{#4}{#5}}
\def \inP #1 < {\@ifnextchar{o}{\inR {#1} < }{\inLaux {#1} < }}
\def \inLaux #1 < #2 | {\@ifnextchar{o}{\inL #1 < #2 | }{\inRP #1 < #2 | }}
\def \inL #1 < #2 | #3 > . #4 {\OrRL{#1}{#2}{\nonconnectable}{#4}}
\def \inR #1 < #2 | #3 > . #4 {\OrRR{#1}{\nonconnectable}{#3}{#4}}
\def \inRP #1 < #2 | #3 > . #4 {\OrRP{#1}{#2}{#3}{#4}}
\def \choice #1 . < #2 #3 | #4 #5 > {\OrL{#1}{#2}{#3}{#4}{#5}}

\def \QED{%
 \unskip
 \hfill 
 \begingroup
 \unitlength\point
 \linethickness{.5\point}%
 \framebox(7,7){}%
 \endgroup
 \kern 10\point
}

 { \end{array} \right . }

\def \same	{\mathrel{\equiv}}

\def \Natural{
{\sf I\kern-1\point N}}

\def \mybullet{\vspace*{2mm} ~ \kern-5mm $\bullet$}

\def \import{\textsl{import}}

\makeatother

\def \ellips(#1,#2){\put (0,0)	{\psellipse[unit=\unitlength, linewidth=1.25pt](#1,#2)}}
